\documentclass[conference]{IEEEtran}
\IEEEoverridecommandlockouts
\usepackage{balance}
\usepackage{float}
\usepackage{amsthm}
\usepackage{amsmath,amssymb,amsfonts}
\usepackage{algorithmic}
\usepackage{graphicx}
\usepackage{subfigure}
\usepackage{textcomp}
\usepackage{xcolor}
\usepackage{accents}
\usepackage{scalerel}
\usepackage{booktabs}
\usepackage{lettrine}
\usepackage{threeparttable}
\usepackage{bbding}
\usepackage{booktabs, multirow, boldline, hhline, array}
\usepackage{verbatim,placeins}

\newtheorem{definition}{Definition}
\newtheorem{theorem}{Theorem}
\newtheorem{lemma}{Lemma}

\newtheorem{proposition}{Proposition}

\def\BibTeX{{\rm B\kern-.05em{\sc i\kern-.025em b}\kern-.08em
    T\kern-.1667em\lower.7ex\hbox{E}\kern-.125emX}}
\begin{document}

\title{\Huge A Cooperative Statistical Approach for Abnormal Node Detection with Adversary Resistance}

\author{
\IEEEauthorblockN{Yingying Huangfu\IEEEauthorrefmark{1} and 
Tian Bai\IEEEauthorrefmark{2}}
\IEEEauthorblockA{\IEEEauthorrefmark{1}Shield Lab, Huawei Singapore Research Center, Singapore, Singapore\\
\IEEEauthorrefmark{2}School of Computing and Data Science, The University of Hong Kong, Hong Kong, China\\
E-mails:~yingying.huangfu@outlook.com, tianbai@hku.hk, \IEEEauthorrefmark{2}\textit{Corresponding author}
}}

\maketitle

\begin{abstract}

Distinguishing abnormal nodes from those with normal packet loss in clusters helps reduce the loss of clustered network resources. 
The detection performance of existing detection schemes is limited by the techniques to quantify node behaviors, and most schemes cannot avoid being misled by the falsified information. 
This paper presents a novel probabilistic abnormal node detection scheme CSD -- Cooperative Statistical Detection -- for accurate and efficient detection in the existence of falsified detection data in clustered networks. 
Specifically, employing the likelihood ratio test (LRT) based detection method to measure node forwarding behaviors, we propose a modified Z-score based falsification-resistant mechanism to filter out falsifications. 
We show that both the false alarm and missed detection probabilities can decrease \textit{exponentially} if and only if the transmissions from the nodes falsifying the data are less than half of the total.
Furthermore, the optimal threshold of the modified Z-score method is derived, which guarantees perfect detection of our CSD under \textit{any} falsification strategy in the proposed detection model.
Evaluation results validate the effectiveness, robustness, and superiority of our scheme compared to the state-of-the-art.

\end{abstract}


\vspace{-0.2cm}
\section{Introduction}

In short-distance communications, including smart homes, multi-agent systems, and health monitoring systems, the performance of packet forwarding is vital for guaranteeing quality of service requirements~\cite{LM23,KF23}.
The network arranged in a clustered topology can simplify the routing process and support various applications, such as data fusion and query processing within the network~\cite{aa21,dc20}.
However, with such great potential, reaping the benefits of the clustered topology still needs to overcome the issues of anomalous forwarding behaviors caused by abnormal nodes, which is common in peer-to-peer network architecture~\cite{sp21, AB21}.

Currently, the main strategy for discovering in-network anomalous behaviors is to implement abnormal node detection during packet forwarding along the routing path.  
In general, the majority of existing routing protocols assume that nodes follow the principles of the protocols entirely acting in altruistic behaviors.
However, due to malevolent or selfish characteristics, abnormal nodes can refuse to forward data packets deliberately, a.k.a.\ routing misbehavior.
Furthermore, in-network abnormal nodes consistently disrupt communication between nodes by dropping transit traffic, resulting in significant packet loss and missing data.
Due to the limited communication resources and hierarchically aggregated data, such abnormal nodes can seriously damage clustered networks.
Besides, knowledge of abnormal nodes is always uncertain, leading to stumbling blocks in establishing security links in distributed networking systems~\cite{at22}.
Therefore, it is of great necessity to identify abnormal nodes with routing misbehavior to secure routing, ensure data integrity, and stabilize network performance.

The problem of abnormal node detection based on packet loss has been studied for decades. 
In recent years, researchers have turned to \textit{statistics-based} approaches focusing on analyzing the behaviors of abnormal nodes. 
Statistical hypothesis testing~\cite{sg06}, as a typical statistical technique, provides theoretical foundations for abnormal detection without the need to modify network protocols or presume the nodes' credibility. 
However, most statistical detection schemes~\cite{gs17,at19, hz22} require massive forwarding observations for analysis, which inevitably increases the cost in low-resource short-distance networks and, as a result, decreases the detection efficiency. 
Recently, the method based on the likelihood ratio test~(LRT) presented in~\cite{HZ21}, which estimates the real abnormal packet loss rate~(PLR) as a default value, was adopted to enhance detection efficiency.
This method proposed a two-phase detection, an effective technique commonly used for identifying abnormal nodes in clustered networks.
The addressed problem is a distributed detection problem where some nodes first make a local decision on the clustered head, and later the detection data of local decisions are aggregated at a trusted node, e.g., the fusion center~(FC), for the overall decision.
In the second phase, malicious nodes could deceive the overall detector and confound its decision on anomaly detection, e.g., by purposely and cooperatively falsifying the detection data to be the opposite. 
In such circumstances, the reliability of the detection data becomes questionable.
In this work, we aim to resolve the limitations of the aforementioned work and propose an integrated \textbf{C}ooperative \textbf{S}tatistical \textbf{D}etection (CSD) scheme to resist the falsification of detection data from malicious detectors. 

\vspace{-0.2cm}
\subsection{State of the Art}

\subsubsection{Regular Abnormal Node Detection}
The authors in~\cite{mg00} described a simple conceptual framework of the trust-based abnormal node detection method from the perspective of packet loss, but did not propose any specific trust algorithm.
In~\cite{rz16}, the detailed trust tables were established at each node to evaluate neighbor forwarding behaviors and integrate trust values to identify abnormal nodes.
A dynamic trust framework was designed in~\cite{wl18} to adopt a weighted summation method to update the trust values.
However, the solutions presented in~\cite{rz16,wl18} cannot accurately characterize the anomaly behavior of abnormal nodes.
Recently, researchers have turned to machine learning based methods that achieved promising detection performance, but their application in the real world was limited due to insufficient training data and led to vast communication overhead~\cite{et18, BT23, ZA23}.
The sequential probability ratio test and Hoeffding test were adopted to effectively detect abnormal nodes under independent and identically distributed~(i.i.d.) traffic~\cite{gs17,hf22}. 
However, they require a large number of behavior samples and take a long time to make decisions.
Then, a two-phase distribution mechanism was proposed to quickly identify abnormal behavior by estimating trust values sent to the trusted node~\cite{kz17}.
Some other works~\cite{at19,HZ21,hf22} followed a similar idea but based on the generalized or the standard LRTs.
However, these solutions focused on identifying anomalies, leaving out the adversarial attack issue that malicious nodes produce falsifications to mislead the detection and significantly reduce detection performance.

\subsubsection{Abnormal Node Detection with Adversarial Strategy}

The latest research considered adversarial scenarios of malicious detectors that falsified detection data.
The authors in~\cite{bf19} modified the trust-based method to guarantee detection performance against malicious nodes by removing untrustworthy data, whereas it cannot effectively distinguish falsified data due to the empirical evaluation mechanism.
The solution in~\cite{ll19} collects the trust value from the neighboring nodes to avoid being misled by falsifications, whereas it did not consider cooperative attacks of the neighboring nodes.
In~\cite{ac18}, the authors proposed an abnormal node detection method that can find falsified data by detecting cooperative misbehavior.
The work of~\cite{rz16} used the standard deviation of trust values to identify falsifications during the detection process.
Despite many benefits in~\cite{ll19,ac18,rz16,bf19}, these solutions did not demonstrate their effectiveness in clustered networks.

\subsubsection{Limitations}

Taken together, the two-phase detection~\cite{kz17,wl18,pr20,at19} was widely used to identify abnormal nodes by aggregating the detection data.
However, most solutions lose effectiveness in adversarial scenarios, which must decrease detection performance.
Furthermore, abnormal node detection schemes that considered the adversarial strategy did not demonstrate their effectiveness in clustered networks.
Unlike existing work, this paper considers the adversarial attacks in the abnormal node detection process for clustered networks where the malicious detectors know the whole detection scheme.
In addition, this work aims to guarantee the \textit{perfect detection} of abnormal node detection in both regular and adversarial scenarios, i.e., the ideal detection for both the false alarm probability~(FAP) and the missed detection probability~(MDP) approach zero exponentially.

\subsection{Contributions}

This work aims to develop a detection approach for the perfect detection of abnormal nodes in clustered networks with i.i.d.\ traffic under an \textit{unknown} distribution, taking into account the adversarial situation where malicious nodes falsify detection information to deceive the detector.
To tackle these issues, this paper presents a novel CSD scheme, which utilizes the designed LRT and the modified Z-score to detect anomalous packet loss and identify falsified detection data, respectively.
The performance of the proposed detection is robust, even if the malicious node has complete knowledge of the detection strategy.
In particular, the proposed CSD improves detection accuracy and efficiency in clustered networks, whether in regular or adversarial environments.
The main contributions of this paper are summarized as follows:

\begin{itemize}
    \item
    We propose a novel CSD to detect anomalies from the perspective of packet loss by aggregating the detection data w.r.t. the observed forwarding behaviors of nodes.
    We embed a modified Z-score method in the CSD to distinguish falsified data from malicious attacks by normalizing the detection data.
    \item
    We provide a comprehensive theoretical analysis with respect to the proposed CSD in detection accuracy and efficiency.
    We show that both the FAP and the MDP of the CSD decay to zero exponentially.
    Furthermore, we derive the optimal removal threshold of the modified Z-score for falsifications with arbitrary cooperative strategies.
    \item
    We evaluate the performance of the proposed CSD through numerical evaluations, achieving extremely high detection accuracy of $99.95\%$ and $99\%$ in regular and adversarial scenarios, respectively.
    The simulation results show that the proposed robust CSD outperforms state-of-the-art methods in efficiency and effectiveness.
\end{itemize}

\section{Preliminaries}

\subsection{Network and Threat Model}\label{sec:pre:model}

In the clustered network, nodes are often grouped into disjoint clusters to aggregate data hierarchically to an FC.
Without loss of generality, consider one cluster with a cluster head $v_{m}$, usually the node with high connectivity, and a set of cluster nodes $\mathcal{C}_{m}$. 
The links in the clustered network are determined by the routes that can be established by the standardized IPv6 routing protocol~\cite{rpl12,op16}.

The traffic from the cluster nodes to the cluster head is upward routing.
Each cluster node transmits a different number of packets to the cluster head through upward routing.
In general, cluster nodes can monitor and distinguish the forwarding packets of their cluster head by the well-known monitoring technique of Watchdog or the Acknowledgement mechanisms~\cite{mg00,tl18,ie16}.
This work assumes that the channel between two nodes is discrete, memoryless, and error-free.
In addition, it is assumed that there exists a side channel from each cluster node to a \textit{trusted node} that is usually the gateway, the border router, or the FC.

In normal transmissions, packet loss occurs at nodes due to their protocol settings and the limited capacity of the receiving queues. 
Following with~\cite{rz16, at19}, this work assumes that the probability of packet loss of a specific node~(e.g., the cluster head) is known as $q_{\mathrm{n}} \in (0, 1)$, different for different nodes. 

Abnormal cluster heads aim to disrupt the communication between the FC and cluster nodes.
Thus, it is assumed that an abnormal cluster head drops the packets transmitted through itself with a higher probability than normal.
The packet loss caused by the abnormal one is caused by the network environment and routing misbehavior independently.
Let $q_{\mathrm{a}}$ be the probability of the \textit{deliberate} packet loss, which is independent of $q_{\mathrm{n}}$.
The total probability of packet loss of an abnormal cluster head is denoted by $q$ and is subject to: 
\begin{equation}
    q=1-(1-q_{\mathrm{n}})(1-q_{\mathrm{a}})=q_{\mathrm{n}}+q_{\mathrm{a}}-q_{\mathrm{n}}q_{\mathrm{a}}.
\label{sys:abnormal rate}
\end{equation}

This work considers the two-phase mechanism that is widely used to detect abnormal nodes in clustered networks~\cite{kz17,wl18,pr20,at19}.
A succinct characterization of the two-phase detection mechanism is as follows: 
\begin{itemize}
    \item \textbf{The local detection phase:}~Each cluster node generates detection data, i.e., the log-likelihood ratio~(LLR) in the LRT, w.r.t. the behavior of the cluster head;
    \item \textbf{The overall detection phase:}~The overall detector, i.e., the trusted node, makes the decision based on the aggregated detection data from the local detection phase.
\end{itemize}

We consider an adversarial scenario as follows. Based on the two-phase detection mechanism, malicious cluster nodes can falsify the detection data to prevent their anomalous cluster head from being detected or frame their normal cluster head as an abnormal node.
In particular, the malicious cluster node knows the identity of its corresponding cluster head through observation.
Fig.~\ref{sys:falsification} illustrates the two considered cases of adversarial attacks.
The solid lines indicate upward routing paths, and the dashed lines denote the monitoring processes of the forwarding packets.

Suppose that the cluster head $v_{m}$ is the node to be checked, and $|\mathcal{C}_{m}|$ is the number of its cluster nodes. 
In \textbf{Case I}, $v_{m}$ is the abnormal node that aims to disrupt communication, and $\mathcal{C}_{m}^{(1)}=\{\,v_1,\ldots,v_j\,\}$ is the set of malicious cluster nodes that generate falsified detection data to claim node $v_{m}$ is normal.
$\mathcal{C}_{m}^{(0)}=\{\,v_{j+1},\ldots,v_{|\mathcal{C}_{m}|}\,\}$ is the set of benign cluster nodes providing true detection data.
In \textbf{Case II}, $v_{m}$ is the normal node with no incentive to disrupt communication, and malicious cluster nodes in $\mathcal{C}_{m}^{(1)}$ generate falsified detection data to frame $v_{m}$ as an abnormal node. 
Benign cluster nodes in $\mathcal{C}_{m}^{(0)}$ claim that node $v_{m}$ is normal.
Details of the specific falsification strategy are introduced in Section~\ref{sec:soulition:z-score}.

\begin{figure}[!tbp]
    \centering
    \includegraphics[height=3.6cm]{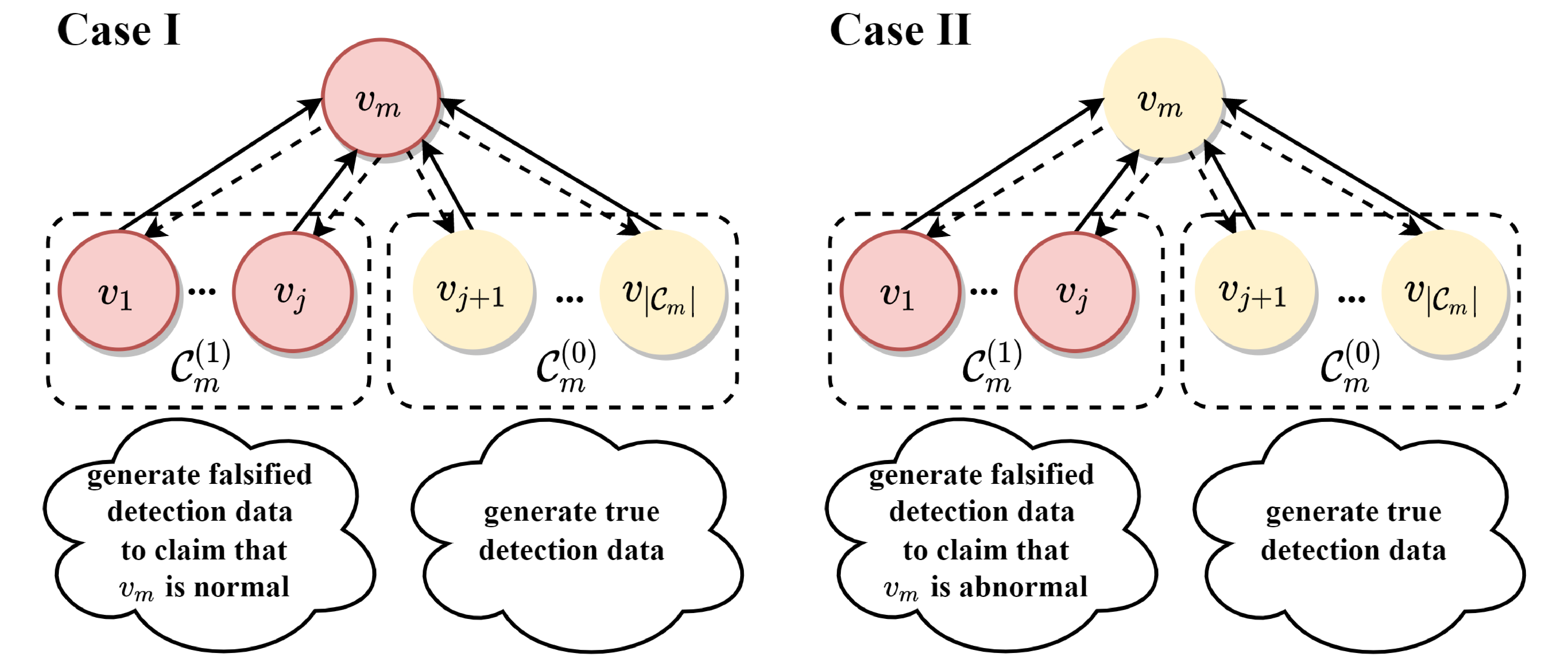}
    \caption{The principle of adversarial attacks. $v_{m}$: the cluster head to be checked if abnormal; $v_1,\ldots,v_j$: malicious cluster nodes; $v_{j+1},\ldots,v_{|\mathcal{C}_{m}|}$: benign cluster nodes.}
    \label{sys:falsification}
\end{figure}

\subsection{Decision Rules of the LRT-based Method}

It can be seen from Section~\ref{sec:pre:model} that the PLR is $q_{\mathrm{n}}$ when the cluster head operates in normal mode. 
However, the PLR $q$ is higher than $q_{\mathrm{n}}$ when the described routing misbehavior occurs. 
Thus, the PLR is a crucial indicator of the presence of the abnormal node.
Due to the difficulty of directly obtaining the knowledge of abnormal nodes, the previous work~\cite{HZ21} uses a detecting PLR $q_{\mathrm{d}}$ which is determined by Eq.~\eqref{sys:abnormal rate} given a possible threshold of $q_{\mathrm{a}}$. 
Let $q_m$ be the PLR of $v_{m}$, then the method performs an LRT of the following binary hypotheses:
\begin{equation}
\begin{aligned} 
    H_0: q_m=q_{\mathrm{n}},
    \quad&&\quad
    H_1: q_m=q_{\mathrm{d}}.
\end{aligned}
\end{equation}

In this work, the duration of the abnormal detection operation is divided into a sequence of detection periods, while each period maintains an independent two-phase detection.

In the $\tau$-th detection period for integer $\tau \in [1, t]$, the number of packets from each cluster node $v_{i} \in \mathcal{C}_{m}$ to $v_{m}$, denoted by $\ell_{i}(\tau)$, is an i.i.d.\ bounded random variable with the mean $\mu_{i}$ and the variance $\sigma^2_i$ for any $i = 1,2,\ldots, |\mathcal{C}_{m}|$.
We notice that our model is a generalization, as the number of packets follows a probabilistic setting.
Let $k_i(\tau)$ be the number of the dropped operation of $v_{m}$ observed from $v_{i}$ in the $\tau$-th period.
Then the decision rule of the LRT in the overall phase is
\begin{equation}
    \Lambda_m(t)= \sum\nolimits_{v_{i} \in \mathcal{C}_{m}} \lambda_{i}(t)~~ \underaccent{H_0}{\accentset{H_1}{\gtrless}} ~~ \gamma,
    \label{CSD:overall decision}
\end{equation}
where $\lambda_{i}(t)$, the LLR generated by $v_{i}$ in the $t$-th detection period in the local phase, can be derived as
\begin{equation}
    \log \left( \frac{1 - q_{\mathrm{d}}}{1-q_{\mathrm{n}}} \right) \sum_{\tau = 1}^t \ell_{i}(\tau) + \log \left( \frac{q_{\mathrm{d}}(1-q_{\mathrm{n}})}{q_{\mathrm{n}}(1-q_{\mathrm{d}})} \right) \sum_{\tau = 1}^t k_{i}(\tau),
    \label{CSD:llr}
\end{equation}
$\Lambda_m(t)$ is the cumulative LLR of $v_{m}$ in the $t$-th period, and $\gamma \in \mathbb{R}$ is the detection threshold under some certain criterion~(e.g., the maximum a posteriori probability criterion).
That is, $H_0$ is rejected if and only if $\Lambda_m(t) \geq \gamma$, which is the standard for identifying abnormal nodes with high PLRs.

\subsection{Problem Definition} \label{Problem}

Following~\cite{ac18,gs17,rz16,at19}, this work uses FAP and MDP to evaluate the detection performance.
In particular, FAP, also known as the \textit{Type I} error, is the probability that a normal node is incorrectly identified as an abnormal node.
The probability of deciding an abnormal node as a normal node is MDP, that is, \textit{Type II} error. 
Denote $\eta_{1,0}(t)$ and $\eta_{0,1}(t)$ as the FAP and the MDP in the $t$-th detection period, respectively, where $t \geq 1$ is the number of detection periods.
This paper considers the perfect detection~\cite{zx21} defined as follows:

\begin{definition}[Perfect Detection] \label{def: perfect dection}
A detection method is said to achieve perfect detection, if both $\eta_{1,0}(t)$ and $\eta_{0,1}(t)$ approach 0 \textit{exponentially} when $t$ tends to infinity, i.e.,
\begin{equation}    
    \begin{aligned} 
    \eta_{1,0}(t)=\exp{(-\mathcal{O}(t))},&& \eta_{0,1}(t)=\exp{(-\mathcal{O}(t))}.
    \end{aligned}
\end{equation}
\end{definition}

If the perfect detection is achieved, for any $\varepsilon > 0$, there exists a positive real number $T$ such that for all $t>T$, $\eta_{0,1}(t),\eta_{1,0}(t) < \varepsilon$ holds. Moreover, the detection efficiency is reflected by the decay rates of $\eta_{0,1}(t)$ and $\eta_{1,0}(t)$, that is, the asymptotic order of the two.

In addition, to theoretically analyze the perfect detection, we introduce the following definition.

\begin{definition}[Critical Detection Point]
Let $k_0$ be the critical value~(i.e., the root of $\Lambda_{m}(t) = \gamma$) of LRT and $\ell$ be the total number of packets in the first $t$ detection periods.
The critical detection point of LRT denoted as $\beta$ is defined as the limitation of $k_0/\ell$ as $\ell$ tends to infinity, i.e.,
\begin{equation}
    \beta = \lim_{\ell \to \infty} \frac{k_0}{\ell} = \log \left( \frac{1-q_{\mathrm{n}}}{1-q_{\mathrm{d}}} \right) \bigg/ \log \left( \frac{q_{\mathrm{d}}(1-q_{\mathrm{n}})}{q_{\mathrm{n}}(1-q_{\mathrm{d}})} \right).
    \label{CSD:critical detection point}
\end{equation}
\end{definition}

One can verify that $q_{\mathrm{n}} < \beta < q_{\mathrm{d}}$. We may have $k_0 > \ell$ for sufficiently small $\ell$.
Here we define $\binom{\ell}{k} = 0$ when $k > \ell$.

The critical detection point characterizes the perfect detection performance of the standard LRT-based detection method~\cite{LG18, Dembo2010}.
Specifically, we have the following results~\cite{HZ21}.

\begin{proposition}\label{proposition:beta}
    In regular scenarios, the perfect detection of the proposed LRT-based detection method is guaranteed if and only if the PLR of the abnormal node ($q_{m}$) exceeds the critical detection point ($\beta$), that is, $q_{m} > \beta$, regardless of the distribution of the number of packets.
\end{proposition}

The solution that satisfies the perfect detection can correctly and rapidly distinguish abnormal nodes from normal ones, which is the goal of resolving the problem stated in this work.
In general, most existing detection schemes do not meet this property anymore because they ignore the discussion on the security aspect, leading to a lack of resistance mechanism to the falsifications.

Additionally, the considered adversarial scenario is where the decision rule of the LRT-based method in Eq.~\eqref{CSD:overall decision} is public, which can be regarded as a classic game problem named the complete information static game~\cite{RG92}.
Generally speaking, this assumption is common when discussing the security and resilience of the scheme because the adversaries will design targeted attack strategies to break security.
That is to say, a malicious cluster node $v_{i}$, to incorrectly determine whether a cluster head is abnormal, can falsify the value of the LLR $\lambda_{i}(t)$ deliberately in each detection period.
The next section will demonstrate the details of the falsification strategy w.r.t. the LRT-based detection model.
Our goal is to design a robust detection scheme when the LRT-based method is known to malicious cluster nodes.

\section{Proposed Solution}

\subsection{The CSD scheme Through Removing Untrustworthy LLRs}\label{sec:soulition:z-score}

\subsubsection{Falsification Strategy}

As introduced in Section~\ref{sec:pre:model}, since the malicious cluster nodes know the identity of their cluster head, they can falsify the LLRs w.r.t. their cluster head $v_{m}$ to mislead the overall detection. 
Each malicious cluster node accomplishes this goal by generating falsified LLR values in a deterministic way.

Formally, let $\mathcal{C}_{m}^{(1)}$ be the set of malicious cluster nodes of $v_{m}$ and $\mathcal{C}_{m}^{(0)}$ be the set of benign ones.
Apparently, we have $\mathcal{C}_{m}^{(0)} \cup \mathcal{C}_{m}^{(1)} = \mathcal{C}_{m}$.
For any malicious cluster node $v_{i} \in \mathcal{C}_{m}^{(1)}$, the falsified LLRs are generated by a falsified number of \textit{total} packets ($\ell'_{i}(\tau)$) and a falsified number of \textit{dropped} packets ($k'_{i}(\tau)$) in the $\tau$-th detection period for integer $\tau \in [1, t]$.
Specifically, a \textit{falsification strategy} of $v_{i}$ can be characterized in the triple $(\kappa_{i}, q'_{i}, q''_{i})$ whose details are as follows:
\begin{itemize}
    \item There exists a positive real $\kappa_{i}$ such that $\ell'_{i}(\tau)= \kappa_{i} \cdot \ell_{i}(\tau)$.
    \item If $v_{m}$ is a normal node with the PLR $q_{\mathrm{n}}$, $v_{i}$ sends LLRs with a fixed PLR $q_{i}'$, i.e., $k'_{i}(\tau)=q'_{i} \cdot \ell'_{i}(\tau)$. 
    \item If $v_{m}$ is abnormal with PLR $q$, $v_{i}$ generates LLRs with a fixed PLR $q''_{i}$, i.e., $k'_{i}(\tau) = q''_{i} \cdot \ell'_{i}(\tau)$. 
\end{itemize}

Let $\mu_{i}$ be the mean of the number of transmitted packets from $v_{i}$ to $v_{m}$ in each detection period for any $i = 1,2,\ldots, |\mathcal{C}_{m}|$. 
The sum of the mean of the packets from $\mathcal{C}_{m}^{(\omega)}$ to $v_{m}$ is $\mu^{(\omega)} = \sum_{v_{i} \in \mathcal{C}_{m}^{(\omega)}} \mu_{i}$ for $\omega \in \{0,1\}$.
The falsification strategy of all malicious cluster nodes can be characterized by a triple $(\kappa, q', q'')$, which we call the \textit{group falsification strategy}, where for every symbol $\diamondsuit \in \{ \kappa, q', q'' \}$, 
\begin{equation}
    \diamondsuit = \frac{\sum\nolimits_{v_{i} \in \mathcal{C}_{m}^{(1)}} \diamondsuit_{i}\mu_{i}}{\sum_{v_{i} \in \mathcal{C}_{m}^{(1)}}\mu_{i}}
    = \frac{\sum\nolimits_{v_{i} \in \mathcal{C}_{m}^{(1)}} \diamondsuit_{i}\mu_{i}}{\mu^{(1)}}
\end{equation}
is the weighted averages given the weight $\mu_{i}$.

Notice that the cooperation of different malicious cluster nodes is allowed in our model.
For instance, they can pre-negotiate a coherent individual strategy for a specific group falsification.

Next, we show that in the presence of the falsified detection data, it would \textit{not} guarantee the perfect detection of the proposed detection method in the overall detection phase.

\begin{lemma}\label{lemma:mu}
    When the PLR of an abnormal node exceeds the critical detection point $\beta$, the perfect detection would not hold if and only if the ratio of $\mu^{(0)}$ to $\mu^{(1)}$ satisfies:
        \begin{equation}
            \frac{\mu^{(0)}}{\mu^{(1)}} \leq \kappa \cdot \max \left\{\frac{q'- \beta}{\beta - q_{\mathrm{n}}}, \frac{\beta - q''}{q - \beta} \right\}.
        \label{CSD:packet ratio}
    \end{equation}   
\end{lemma}

\begin{proof}
    According to Propositon~\ref{proposition:beta}, it is easy to show that if the PLR of the abnormal node ($q$) does not exceed $\beta$, the detection result is incorrect in the overall detection phase even when all the detection data are true.
    Similarly, if the PLR of the normal node ($q_{\mathrm{n}}$) is no less than $\beta$, the detection result is also incorrect because the FAP $\eta_{1,0}$ will not approach zero.
    By Definition~\ref{def: perfect dection}, the detection result is correct if and only if both the FAP $\eta_{1,0}$ and MDP $\eta_{0,1}$ all exponentially approach zero as the detection periods tend to infinity.
    Hence, when there exist some malicious detectors with the sum of the mean of the packets $\mu^{(1)}$, the detection result in the overall detection phase of the CSD is incorrect if $\mu^{(1)}$ and $\mu^{(0)}$ satisfy
    \begin{equation}
        \frac{\mu^{(0)}q_{\mathrm{n}} + \kappa\mu^{(1)}q'}{\mu^{(0)}+\kappa\mu^{(1)}} \geq \beta ~~~\text{or}~~~\frac{\mu^{(0)}q + \kappa\mu^{(1)}q''}{\mu^{(0)}+\kappa\mu^{(1)}} \leq \beta.
    \end{equation}
    By the relation $q_{\mathrm{n}} < \beta < q_{\mathrm{d}}$, we finally obtain Eq.~\eqref{CSD:packet ratio}.
\end{proof}

We remark that Lemma~\ref{lemma:mu} shows whether the LRT-based method admits the perfect detection only depends on the group falsification strategy, even if some malicious cluster nodes produce falsified LLRs.

\subsubsection{Modified Z-Score Method}

From Lemma~\ref{lemma:mu}, the LRT-based method \textit{cannot} distinguish anomalies from normal nodes correctly in the overall detection phase if the number of packets from malicious cluster nodes satisfies Eq.~\eqref{CSD:packet ratio}. This might happen if: 1) $|\mathcal{C}_{m}^{(1)}| $ is big enough even if each $\mu_{i}$ is small for $v_{i} \in \mathcal{C}_{m}^{(1)}$; 2) $|\mathcal{C}_{m}^{(1)}|$ is small but there is at least one sufficiently large $\mu_{i}$; or 3) $\kappa$ is large enough even if $\mu^{(1)}$ is small. Therefore, finding the falsified LLR values generated from $\mathcal{C}_{m}^{(1)}$ and removing them is indispensable to guarantee the perfect detection of the detection method in the overall detection phase.

Since the LLR value $\lambda_{i}(t)$ from the local detection phase depends on the number of transmitted packets that is different for each cluster node, we define the \textit{normalized} LLR as the LLR of one packet, denoted as
\begin{equation}
    x_{i}(t) =  \frac{\lambda_{i}(t)}{\mu_{i}t},
\end{equation}
where $\mu_{i}$ represents the mean of the number of packets transmitted from $v_{i}$ to its corresponding cluster head, which is known by the trusted node through network measurement. 

It is worth noting that the normalized LLRs of the true LLRs are generated by the same probability distribution, whereas those from falsified ones are not.
An LLR may be a falsified detection data generated from a malicious cluster node if the normalized LLR is an outlier in the group of them denoted by $\{\,x_1(t), x_2(t), \ldots, x_{|\mathcal{C}_{m}(t)|}\,\}$.
Hence, the analysis of the falsified LLRs is converted to decide whether different normalized LLRs belong to the same population and find out the outliers in that population. 

Z-score~\cite{bt16,dg20} is a numerical measure that describes the relationship of a value to the mean of a set of values, which is also called the standard score because it allows the comparison of scores of different kinds of variables through a standardized distribution.
Since the LLR aggregation mechanism at the trusted node computed by Eq.~\eqref{CSD:overall decision} is related to each LLR value, the numbers of the data packets directly determine the result of decisions. 
Considering the effect on the decision by the number of packets, we utilize a modified Z-score to measure $x_{i}(t)$, where the simple arithmetic mean and variance are replaced by the weighted mean and variance of the group:
\begin{equation}
    \left\{
    \begin{aligned}
   \bar{x}(t)
   &= \sum\nolimits_{i=1}^{|\mathcal{C}_{m}|}\mu_{i} x_{i}(t)/\sum\nolimits_{i=1}^{|\mathcal{C}_{m}|} \mu_{i}; \\
    s^2(t)
    &=\sum\nolimits_{i=1}^{|\mathcal{C}_{m}|}\mu_{i}(x_{i}(t)-\bar{x}(t))^2/\sum\nolimits_{i=1}^{|\mathcal{C}_{m}|}\mu_{i}.
    \end{aligned}
    \right.
    \label{va} 
\end{equation}
Hence, the modified Z-score of a normalized LLR $x_{i}(t)$ is
\begin{equation}
    z_i(t) = (x_{i}(t) - \bar{x}(t))/s(t),
\end{equation}
where $\bar{x}(t)$ and $s(t)$ are respectively the weighted mean and the standard deviation of the group of normalized LLRs. 
Note that Z-scores may be positive/negative, indicating the score is above/below the mean.
Given a removal threshold $z_{\mathbf{thr}} \in \mathbb{R}^{+}$, if $|z_i(t)| > z_{\mathbf{thr}}$, $x_{i}(t)$ is an outlier.
Further, $\lambda_{i}(t)$ is untrustworthy to be removed in the overall detection phase.

\subsubsection{Decision Rules of the CSD}

When the trusted node receives the LLRs in each detection period, the modified Z-score method is employed to mitigate the effect of the falsified LLRs when detecting abnormal nodes.
The trusted node makes use of the remaining LLRs after the modified Z-score method to implement LRT in the overall detection phase.
In this way, the CSD is able to resist falsified detection data generated by malicious cluster nodes.

Let $\tilde{\mathcal{C}}_{m}^{(1)} \triangleq \{ v_{i} \in \mathcal{C}_{m}^{(1)} : |z_{i}| \leq z_{\mathbf{thr}} \}$ and $\tilde{\mathcal{C}}_{m}^{(1)} \triangleq \{ v_{i} \in \mathcal{C}_{m}^{(1)} : |z_{i}| \leq z_{\mathbf{thr}} \}$ denote the sets of benign and malicious cluster nodes generating trustworthy LLRs, respectively.
Then the decision rule of the CSD scheme is
\begin{equation}
    \tilde{\Lambda}_m(t)
    = \sum\nolimits_{v_{i} \in \tilde{\mathcal{C}}_{m}^{(0)} \cup \tilde{\mathcal{C}}_{m}^{(1)}} \lambda_{i}(t)~~ 
    \underaccent{H_0}{\accentset{H_1}{\gtrless}} ~~ \gamma,
    \label{CSD:removed detection rule}
\end{equation}
where $\tilde{\Lambda}_m(t)$ is the cumulative LLR of $v_{m}$ in the $t$-th detection period without untrustworthy LLRs and $\gamma$ is the detection threshold the same as Eq.~\eqref{CSD:overall decision}.
It follows that the null hypothesis is rejected if and only if $\tilde{\Lambda}_m(t) \geq \gamma$, which is the rule to identify abnormal nodes with high packet loss. 

\subsection{Theoretical Analysis}

For the clustered networks deployed in a hostile environment, each malicious cluster node $v_{i}$ chooses a particular falsification strategy $(\kappa_{i}, q'_{i}, q''_{i})$ to confound the trusted node, causing an incorrect detection.
The group falsification strategy $(\kappa, q', q'')$ is totally determined by the falsification strategies of all malicious cluster nodes.
To theoretically analyze the perfect detection of the CSD, we first analyze the asymptotic behavior of the parameters in the modified Z-score method.

In this subsection, we always denote $\mu^{(0)}$ and $\mu^{(1)}$ as the sum of the means of the packets from normal and abnormal detectors, respectively.
Let $A = (1 - q_{\mathrm{d}})/(1 - q_{\mathrm{n}})$ and $B = q_{\mathrm{d}}/q_{\mathrm{n}} \cdot A^{-1}$ be two constant.
Then we have
\begin{equation}
    \lambda_{i}(t) = A \cdot \sum\nolimits_{\tau = 1}^{t} \ell_{i}(\tau) + B \cdot \sum\nolimits_{\tau = 1}^{t} k_{i}(\tau).
\end{equation}

To theoretically analyze the detection performance of our CSD, we define the supremum packet ratio with falsified data based on Lemma~\ref{lemma:mu}.

Recall that $\mu^{(0)}$ and $\mu^{(1)}$ are the sum of the means of the packets from normal and abnormal detectors, respectively.
Let $\tilde{\mu}^{(0)} = \sum_{v_{i} \in \tilde{C}_{m}^{(0)}} \mu_{i}$ and $\tilde{\mu}^{(1)} = \sum_{v_{i} \in \tilde{C}_{m}^{(1)}} \mu_{i}$ be the sum of means of the packets transmitted from normal and abnormal detectors who generate trustworthy LLRs, respectively.
We will use $\tilde{\mu}^{(0)}$ and $\tilde{\mu}^{(1)}$ capture the ratio of $\mu^{(0)}$ to $\mu^{(1)}$ such that our proposed scheme admits the perfect detection.
    
\begin{definition}[Supremum Packet Ratio with Falsified Data]
    The supremum packet ratio with falsified data denoted by $\xi$ is defined as the supremum ratio of $\mu^{(0)}$ to $\mu^{(1)}$ for perfect detection in any group falsification strategy, i.e.,
    \begin{equation}
        \xi = \xi \left(z_\mathbf{thr} \right) =
        \inf\nolimits_{\kappa, q', q''}  \sup\nolimits_{\mu^{(0)},\mu^{(1)}}\Xi,
    \end{equation}
    where
    \begin{equation}
        \Xi =
        \left\{ \frac{\mu^{(0)}}{\mu^{(1)}} : \frac{\tilde{\mu}^{(0)}}{\tilde{\mu}^{(1)}} > \kappa \cdot \max \left\{ \frac{q'- \beta}{\beta - q_{\mathrm{n}}}, \frac{\beta - q''}{q - \beta} \right\} \right\}.
        \label{CSD:condition}
    \end{equation}
\end{definition}

Now we investigate the detection accuracy of the CSD with the removal threshold $z_{\mathbf{thr}} \in \mathbb{R}^+ $.
First, we give an upper bound of $\xi$.

\begin{lemma}\label{lemma:necessity}
    Given a $z_{\mathbf{thr}} \in \mathbb{R}^+ $, let $\xi^{*}= \min\{\, z_{\mathbf{thr}}^{-2}, z_{\mathbf{thr}}^{2} \,\}$.
    When $\mu^{(0)}/\mu^{(1)} \in [0, \xi^{*}) \cup [0, 1]$, there exists some group falsification strategies $(\kappa, q', q'')$ such that CSD cannot realize the perfect detection, which leads to $ \xi \leq \xi^{*}$.
\end{lemma}

\begin{proof}    
    Let $\hat{q} = q'$ if $v_{m}$ is normal and $\hat{q} = q''$ if $v_{m}$ is abnormal.
    According to Borel Law of Large Number~\cite{bo09}, we can derive that $
        \lim_{t \to \infty}\bar{x}(t) 
        = \mathbb{E}[\bar{x}(t)] 
        = \frac{\mu^{(0)}(A+Bq_{\mathrm{n}})+\mu^{(1)}(A+B\hat{q})\kappa}{\mu^{(0)}+\mu^{(1)}}$,
    and
        $
        \lim_{t \to \infty} s(t) 
        = \mathbb{E}[s(t)] 
        = \frac{|A(1 - \kappa) + B(q - \hat{q}\kappa)| \cdot \sqrt{\mu^{(0)}\mu^{(1)}}}{\mu^{(0)}+\mu^{(1)}}$.

    According to the definition of Z-score, one can easily find that $\lim_{t \to \infty} |z_{i}(t)| = \sqrt{\mu^{(1)}/\mu^{(0)}}$ when $v_{i} \in \mathcal{C}_{m}^{(0)}$; and $\lim_{t \to \infty} |z_{i}(t)| = \sqrt{\mu^{(0)}/\mu^{(1)}}$ when $v_{i} \in \mathcal{C}_{m}^{(1)}$.
    It follows by the criteria for removing LLR outliers when $t \to \infty$:
    \begin{equation}
        \left( \tilde{\mu}^{(0)}, \tilde{\mu}^{(1)} \right) = \left\{ 
        \begin{aligned}
            &(\mu^{(0)}, \mu^{(1)} ),   && z_{\mathbf{thr}}^{-2} \leq \frac{\mu^{(0)}}{\mu^{(1)}} \leq z_{\mathbf{thr}}^{2};   \\
            &( \mu^{(0)}, 0 ),          && \frac{\mu^{(0)}}{\mu^{(1)}} \geq z_{\mathbf{thr}}^{-2}, \frac{\mu^{(0)}}{\mu^{(1)}} > z_{\mathbf{thr}}^{2};  \\
            &(0, \mu^{(1)} ),           && \frac{\mu^{(0)}}{\mu^{(1)}} < z_{\mathbf{thr}}^{-2}, \frac{\mu^{(0)}}{\mu^{(1)}} \leq z_{\mathbf{thr}}^{2};  \\
            &(0,0 ),                    && z_{\mathbf{thr}}^{-2} > \frac{\mu^{(0)}}{\mu^{(1)}} > z_{\mathbf{thr}}^{2}.  \\
        \end{aligned} \right.
        \label{CSD:removed mu}
    \end{equation}

    There are two cases of $z_{\mathbf{thr}}$ to be considered when $\mu^{(0)}/\mu^{(1)} \in [0, \xi^{*}) \cup [0, 1]$, where $\xi^{*}= \min\{\, z_{\mathbf{thr}}^{-2}, z_{\mathbf{thr}}^{2} \,\}$. 
    
    Case 1: When $z_{\mathbf{thr}} < 1$, it holds that $\mu^{(0)}/\mu^{(1)} < z_{\mathbf{thr}}^{-2}$, which leads to $\tilde{\mu}^{(0)} = 0$.
    Hence, the detection results would be incorrect because all the LLRs are falsified.
    
    Case 2: When $z_{\mathbf{thr}} \geq 1$, it holds that $\mu^{(0)}/\mu^{(1)} \leq z_{\mathbf{thr}}^{2}$.
    If $\mu^{(0)}/\mu^{(1)} < z_{\mathbf{thr}}^{-2}$, then $\tilde{\mu}^{(0)} = 0$ holds; and if $z_{\mathbf{thr}}^{2} \geq \mu^{(0)}/\mu^{(1)} \geq z_{\mathbf{thr}}^{-2}$, then $\tilde{\mu}^{(0)} = \mu^{(0)}$ and $\tilde{\mu}^{(1)} = \mu^{(1)}$ hold.
    When $\kappa$ is sufficiently large, the value $\mu^{(0)}/\mu^{(1)}$ could be less than $\kappa \cdot \max \left\{ \frac{q'-\beta}{\beta - q_{\mathrm{n}}}, \frac{\beta-q''}{q - \beta} \right\}$, which implies that the LRT-based detection results is incorrect under this condition by Lemma~\ref{lemma:mu}.
    
    Therefore, if $\mu^{(0)}/\mu^{(1)} \in [0, \xi^{*}) \cup [0, 1]$, CSD would fail to detect anomalies for fixed $z_{\mathbf{thr}} \in \mathbb{R}^{+}$, where $\xi^{*}=\xi^{*}(z_{\mathbf{thr}}) = \min\{\, z_{\mathbf{thr}}^{-2}, z_{\mathbf{thr}}^{2} \,\}$.
\end{proof}

Next, we give the lower bound of $\xi$.

\begin{lemma}\label{lemma:sufficiency}
    Given a $z_{\mathbf{thr}} \in \mathbb{R}^+ $, let $\xi^{*} = \min\{\, z_{\mathbf{thr}}^{-2}, z_{\mathbf{thr}}^{2} \,\}$.
    When the PLR of abnormal nodes exceeds the critical detection point $\beta$, the CSD can realize perfect detection if $\mu^{(0)}/\mu^{(1)} \in (\xi^{*}, +\infty) \cap (1, +\infty)$, which leads to $\xi \geq \xi^{*}$.
\end{lemma}

\begin{proof}
    We will prove that $\tilde{\mu}^{(1)} = 0$ which implies that the detection result of CSD is correct by Lemma~\ref{lemma:mu}.
    We also consider two cases of $z_{\mathbf{thr}}$ in eq.~\eqref{CSD:removed mu} when $\mu^{(0)}/\mu^{(1)} \in [\xi^{*}, +\infty) \cap (1, +\infty)$ where $\xi^{*} = \min\{\, z_{\mathbf{thr}}^{-2}, z_{\mathbf{thr}}^{2} \,\}$:
    (i) When $z_{\mathbf{thr}} \geq 1$, it holds that $\mu^{(0)}/\mu^{(1)} > z_{\mathbf{thr}}^{2} > 1$, which means that $\tilde{\mu}^{(1)} = 0$.
    (ii) When $z_{\mathbf{thr}} < 1$, it holds that $\mu^{(0)}/\mu^{(1)} \geq z_{\mathbf{thr}}^{2}$.
    If $\mu^{(0)}/\mu^{(1)} > z_{\mathbf{thr}}^{-2}$, then $\tilde{\mu}^{(1)} = 0$;
    if $\mu^{(0)}/\mu^{(1)} = z_{\mathbf{thr}}^{-2} > 1$, it also holds that $\tilde{\mu}^{(1)} = 0$.    
    Therefore, we have $ \xi \geq \xi^{*}$, where $\xi^{*}=\xi^{*}(z_{\mathbf{thr}}) = \min\{\, z_{\mathbf{thr}}^{-2}, z_{\mathbf{thr}}^{2} \,\}$.
\end{proof}

The unique optimal removal threshold is now derived.

\begin{theorem}
    The optimal removal threshold is $z_{\mathbf{thr}} = 1$, which guarantees the perfect detection of CSD scheme if and only if the PLR of abnormal nodes exceeds the critical detection point and the sum of means of packets transmitted from the malicious cluster nodes is less than half of the total sum of means of packets from all cluster nodes, i.e.,$q>\beta$ and $\mu^{(0)}/\mu^{(1)} > 1$.
    \label{CSD:theorem3}
\end{theorem}

\begin{proof}
    We have an upper bound on $\xi$ by Lemma~\ref{lemma:necessity} that $\xi \leq \xi^{*}$ and a lower bound by Lemma~\ref{lemma:sufficiency} that $\xi \geq \xi^{*}$. Therefore, it follows that $\xi = \xi^{*}$.
    Then we finally get $\sup\nolimits_{z_{\mathbf{thr}} \in \mathbb{R}^{+}} \xi (z_{\mathbf{thr}}) = 1$, which means that the supremum packet ratio with falsified data $\xi = 1$ if and only if $z_{\mathbf{thr}} = 1$.
    Note that $\xi = 1$ is equivalent to $\mu^{(0)}/\mu^{(1)} > 1$.
\end{proof}

It is worth noting that the network model considers the scenario of different $\mu_{i}$ for $i = 1,2,\ldots,|C_m|$.
A particular case of the network model is that all $\mu_{i}$ are equal.
In this case, $\mu^{(1)}/\mu^{(0)}$ is equivalent to the ratio of the number of abnormal nodes to that of normal nodes.

Finally, we show the perfect detection of the CSD with optimal threshold when the sum of means of packets transmitted from the malicious cluster nodes is less than half of the total sum of means of packets from all cluster nodes.

\begin{theorem}\label{theorem:perfect}
    When the removal threshold is optimal (i.e., $z_{\mathbf{thr}} = 1$), $q_{\mathrm{n}} < \beta < q$, and $\mu^{(0)}/\mu^{(1)} > 1$, the CSD admits the perfect detection.
\end{theorem}

\begin{proof}
    Recall that $\ell_{i}(\tau)$ is a bounded random variable with the mean $\mu_{i}$.
    By Hoeffding's inequality, for any small positive $\varepsilon_{0}$, one can easily find that $
        \Pr \left\{\, |\lambda_{i}(t) - \mathbb{E}[\lambda_{i}(t)]| \geq \varepsilon_{0} \mu_{i} t \,\right\}
        \leq \exp \left( -\mathcal{O}(\varepsilon_{0}t) \right)$.
    Consequently, for any bounded random variable $X(t) \in \{\, x_{i}(t), \bar{x}(t), s(t) \,\}$, it holds $
        \Pr \left\{\, |X(t) - \mathbb{E}[X(t)]| \geq \varepsilon_{0} \,\right\}
        \leq \exp \left( -\mathcal{O}(\varepsilon_{0}t) \right)$.
    Moreover, according to the definition of Z-score, for any small positive $\varepsilon_{0} > 0$, we have $
            \Pr \left\{\, |z_{i}(t) - \mathbb{E}[z_{i}(t)]| \geq \varepsilon \,\right\}
            \leq 1 - (1 - \exp(-\mathcal{O}(\varepsilon_{0} t))^{3} = \exp \left( -\mathcal{O}(\varepsilon_{0}t) \right)$,
    where
    $\varepsilon =\varepsilon_{0} \cdot \mathbb{E}[2s(t) + x_{i}(t) - \bar{x}(t)]/(\mathbb{E}[s(t)](\mathbb{E}[s(t)] - \varepsilon_{0}))$.    
    By the condition $\mu^{(0)} > \mu^{(1)}$, there exists a small enough constant $\varepsilon_{0} > 0$ satisfying that 
    $\sqrt{\mu^{(1)}/\mu^{(0)}} + \varepsilon < 1 < \sqrt{\mu^{(0)}/\mu^{(1)}} - \varepsilon$.
    
    Consider that $v_{m}$ is a normal cluster node.    
    If $v_{i} \in C_{m}^{(0)}$, we know $q_{i} = q_{\mathrm{n}} < \beta$, and $\mathbb{E}[z_{i}] = \sqrt{\mu^{(1)}/\mu^{(0)}}$.
    Thus, $|z_{i}| > z_{\mathbf{zhr}}$ holds with possibility $\exp(-\mathcal{O}(t))$.
    If $v_{i} \in C_{m}^{(1)}$ and $q_{i} = q'_{i}$, and thus $\mathbb{E}[z_{i}] = \sqrt{\mu^{(1)}/\mu^{(0)}}$.
    Thus, $|z_{i}| \leq z_{\mathbf{zhr}}$ holds with possibility $\exp(-\mathcal{O}(t))$.

    Therefore, we conclude that with the possibility
    $
        1 - (1 - \exp(-\mathcal{O}(t)))^{|\mathcal{C}_{m}|} = |\mathcal{C}_{m}| \cdot \exp(-\mathcal{O}(t)),
    $
    the outliers are exactly all falsified LLRs generated by malicious ones.
    This leads that $\eta_{0, 1}$ approaches $0$ exponentially.
    
    Similarly, if $v_{m}$ is an abnormal cluster node, we can also obtain that the outliers are exactly all falsified LLRs generated by malicious cluster nodes with the possibility $|\mathcal{C}_{m}| \cdot \exp(-\mathcal{O}(t))$.
    It follows that $\eta_{1, 0}$ approaches $0$ exponentially.
\end{proof}

Theorem~\ref{theorem:perfect} shows that the FAP and the MDP decrease \textit{exponentially} as the number of detection periods increases, which indicates that the CSD scheme is potentially capable of achieving high detection accuracy and efficiency with sufficient detection periods, i.e., the perfect detection. 

\section{Performance Evaluation}

This section evaluates the performance of the proposed CSD scheme by the number of $10^6$ Monte-Carlo simulations in each detection period.
Without loss of generality, the simulation consists of one trusted node and one cluster in the clustered network following the settings in previous related works~\cite{ac18,at19, tl18}.
In particular, the cluster has one cluster head $v_{m}$ with eight cluster nodes, i.e., $|\mathcal{C}_{m}|=8$, as depicted in Fig.~\ref{sys:falsification}.
The PLR of a normal $v_{m}$ is set as $0.15$, i.e., $q_{\mathrm{n}} = 0.15$, and the detecting PLR of an abnormal $v_{m}$ is set as $0.2$, i.e., $q_{\mathrm{d}} = 0.2$ ($q_{\mathrm{a}} = 0.0588$).
The numbers of packets are i.i.d.\ sampled from the Poisson distributions for each detection period, and the simulated parameters on eight cluster nodes are given in Table \ref{res:parameters}, where $\mu_{i}$ for integer $i \in [1,8]$ represents both the mean and the variance for the Poisson distributions. 
The detection threshold of the decision rule in Eq.~\eqref{CSD:removed detection rule} is set as $1.4$, e.g., $\gamma = 1.4$. 

Specifically, the following two aspects of the proposed CSD are investigated:
\begin{enumerate}
\item The performance of the CSD scheme compared with two well-known methods~\cite{rz16,at19} under \texttt{Scenario 1}, i.e., in the absence of malicious cluster nodes, in Table~\ref{res:distributions}.

\item The performance of the CSD scheme in the presence of malicious cluster nodes falsifying detection data, i.e., \texttt{Scenario 2-4} in Table~\ref{res:distributions}.
\end{enumerate}

\begin{table}[!t]
    \centering
    \vspace{0.4em}
    \caption{Parameters of the Eight Cluster Nodes}
    \begin{tabular}{p{1cm}<{\centering}|p{0.4cm}<{\centering}p{0.4cm}<{\centering}p{0.4cm}<{\centering}p{0.4cm}<{\centering}p{0.4cm}<{\centering}p{0.4cm}<{\centering}p{0.4cm}<{\centering}p{0.4cm}<{\centering}}
    \toprule 
        node        & $v_1$ & $v_2$ & $v_3$ & $v_4$ & $v_5$ & $v_6$ & $v_7$ & $v_8$ \\
    \midrule
        $\mu_{i}$ & 10    & 12    & 8    & 7    & 9    & 4    & 6    & 15    \\ 
    \bottomrule  
    \end{tabular}
    \label{res:parameters}
\end{table}

\begin{table}[!t]
\renewcommand{\arraystretch}{1.3}
\centering
\caption{\label{res:distributions} Distributions of cluster Nodes}

\begin{tabular}{p{2cm}<{\centering}|p{3.3cm}<{\centering}p{2.2cm}<{\centering}}
    \toprule  
        & Benign & Malicious   \\
    \midrule
        Scenario 1 & $\{ \, v_1,v_2,v_3,v_4,v_5,v_6,v_7,v_8 \, \}$ & $\varnothing$\\
        Scenario 2 & $\{ \,v_4,v_5,v_6,v_7,v_8\, \}$ & $\{ \,v_1,v_2,v_3\, \}$\\
        Scenario 3 & $\{ \,v_4,v_5,v_7,v_8\, \}$ & $\{ \,v_1,v_2,v_3,v_6\, \}$\\
        Scenario 4 & $\{ \,v_4,v_5,v_6,v_8\, \}$ & $\{ \,v_1,v_2,v_3,v_7\, \}$\\
    \bottomrule  
    \end{tabular} 
\vspace{-0.3cm}
\end{table}

\vspace{-0.3cm}
\subsection{Performance in Regular Scenarios}

This subsection investigates the performance of the CSD scheme in the overall detection phase in the regular scenario under \texttt{Scenario 1} in Table~\ref{res:distributions}.
We compare the proposed CSD scheme with two well-known methods, including a trust-based scheme~(TBS)~\cite{rz16} and a statistics-based scheme~(SBS)~\cite{at19}. 
The parameter settings of SBS are the same as those of the proposed CSD scheme.
For the TBS, the parameters follow the settings in~\cite{rz16}:

\begin{itemize}
    \item The range of the trust value of a node is $[0, 200]$ and the initial trust is set as $100$.
    \item The trust value is increased by $1$ if the observed PLR of the node is less than $q_{\mathrm{n}}$; if the observed PLR of the node exceeds $q_{\mathrm{d}}$, it is decreased by $10$; otherwise, it is decreased by $1$.
\end{itemize}

\begin{figure}[!tbp]
    \centering
    \subfigure[$ q=0.185 $]{\label{res:comparison:a}\includegraphics[height=2.7cm]{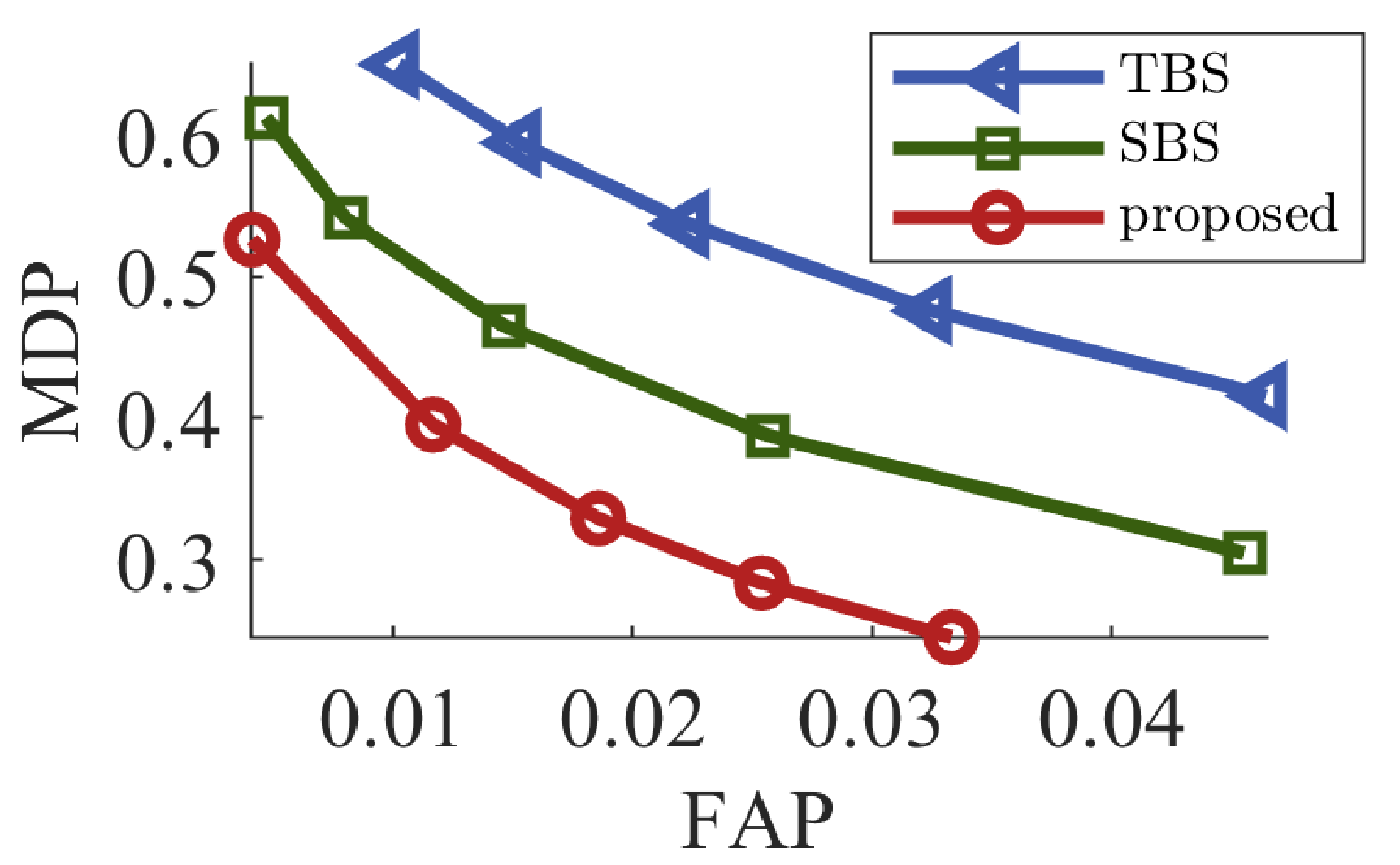}}
    \subfigure[$ q=0.215 $]{\label{res:comparison:b}\includegraphics[height=2.7cm]{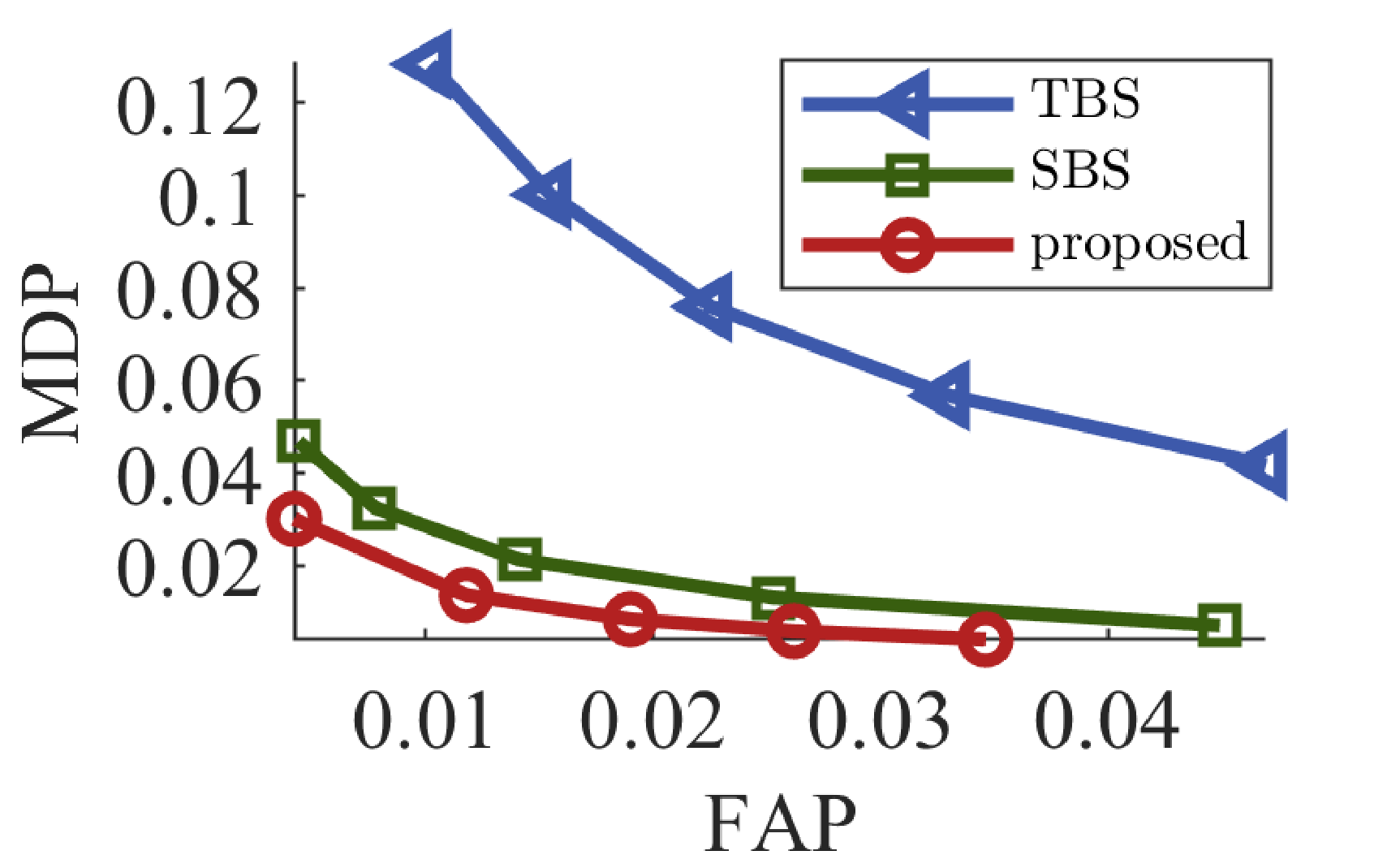}}
    \caption{FAP vs. MDP.}
    \vspace{-0.3cm}
    \label{res:comparison}
\end{figure}

Fig.~\ref{res:comparison} plots the comparison results for $t = 10$.
Upon $q_{\mathrm{n}} = 0.15$ and $q_{\mathrm{d}} = 0.2$, the CSD consistently outperforms the two baseline methods in terms of the MDP given the FAP, demonstrating the superiority of the proposed method in detecting abnormal nodes under various conditions.
Generally, the TBS is inefficient in abnormal node detection because it relies on empirical trust value update rules that cannot be changed dynamically according to the degree of abnormality.
In the SBS, the abnormal PLR is estimated through maximum likelihood estimation, which ignores prior network measurements. 
In contrast, the proposed CSD takes into account the detecting PLR of abnormal nodes, which can be directly estimated by measuring the network transmissions. 
Therefore, the proposed CSD is able to efficiently detect the abnormal node within a short detection time.
Moreover, it can be seen from the curves that the performance gap between the proposed method and the SBS in Fig.~\ref{res:comparison:a} is more evident than that in Fig.~\ref{res:comparison:b}.
As a result of the theoretical guarantees, the proposed method yields significant improvements over the SBS if the PLR of abnormal nodes is close to the normal value and exceeds the critical detection point.

Additionally, in the scenario without malicious cluster nodes, it is evident that the performance of the CSD does not surpass that of the standard LRT-based method.
However, our numerical simulation results show that the FAP and the MDP of both the LRT-based method and the CSD remain almost consistent in each period.
This implies that when no malicious cluster nodes exist, the modified Z-score method rarely identifies any outliers.
Since the curves of the FAP and the MDP of these two methods coincide, the simulation results for the LRT-based method do not appear in Fig.~\ref{res:comparison}.

\vspace{-0.2cm}
\subsection{Performance in Adversarial Scenarios}

The above has shown the effectiveness of the CSD in detecting abnormal nodes in the scenarios with true detection data.
However, the malicious cluster nodes might modify the detection data against the considered model in adversarial scenarios.
Now an evaluation is conducted to investigate the performance of the proposed CSD scheme in resisting the falsified detection data in the overall detection phase.
Specifically, three scenarios shown in Table~\ref{res:distributions}, namely \texttt{Scenario 2-4}, are considered, where the malicious cluster nodes send falsified detection data with respect to its cluster head $v_{m}$ to mislead the trusted node. 

\begin{figure}[!tbp]
\centering
\subfigure[FAP (Scenorio 2)]{\label{res:root:a}\includegraphics[height=2.7cm]{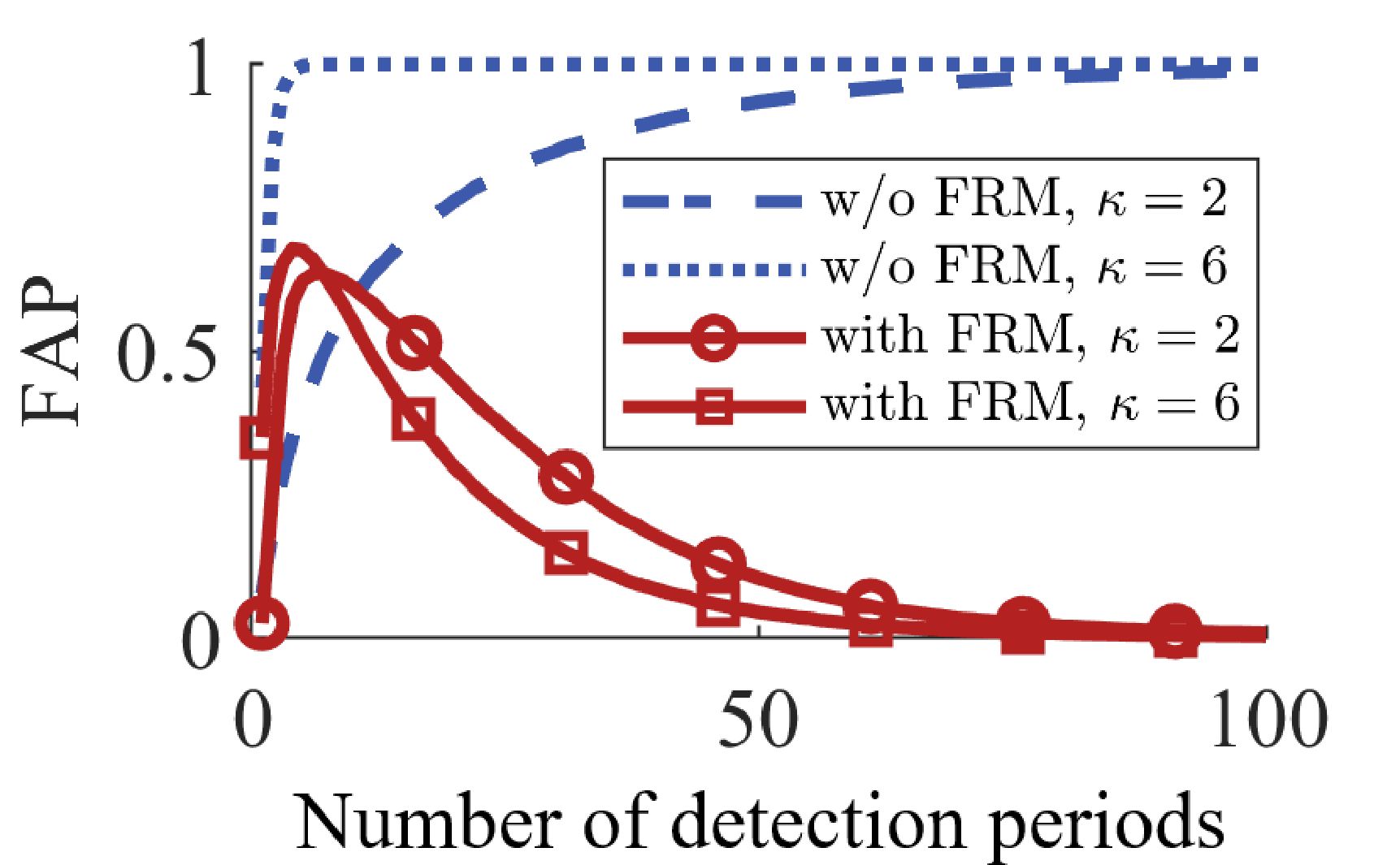}}
\subfigure[MDP (Scenorio 2)]{\label{res:root:b}\includegraphics[height=2.7cm]{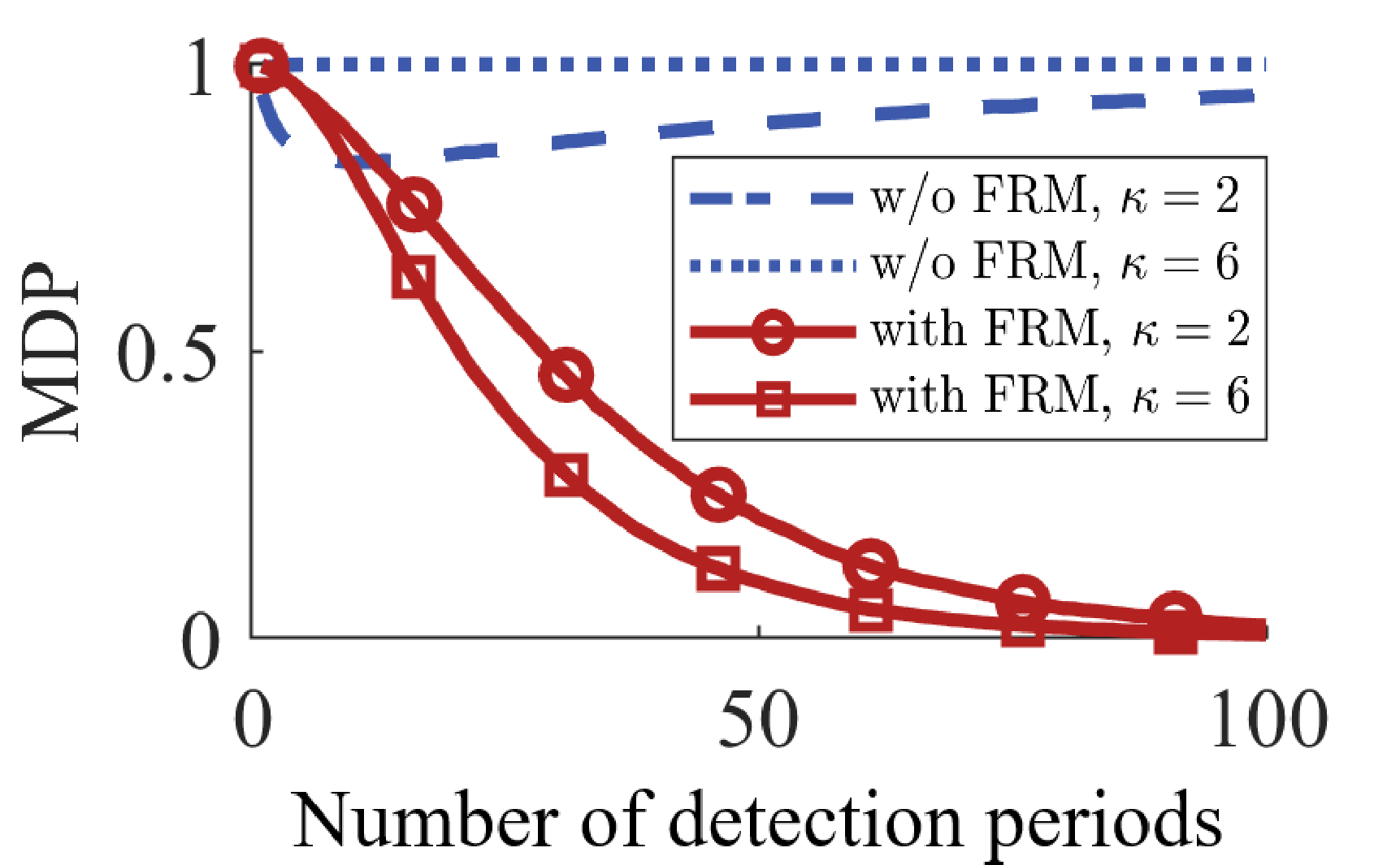}}
\subfigure[FAP (Scenorio 3)]{\label{res:root:c}\includegraphics[height=2.7cm]{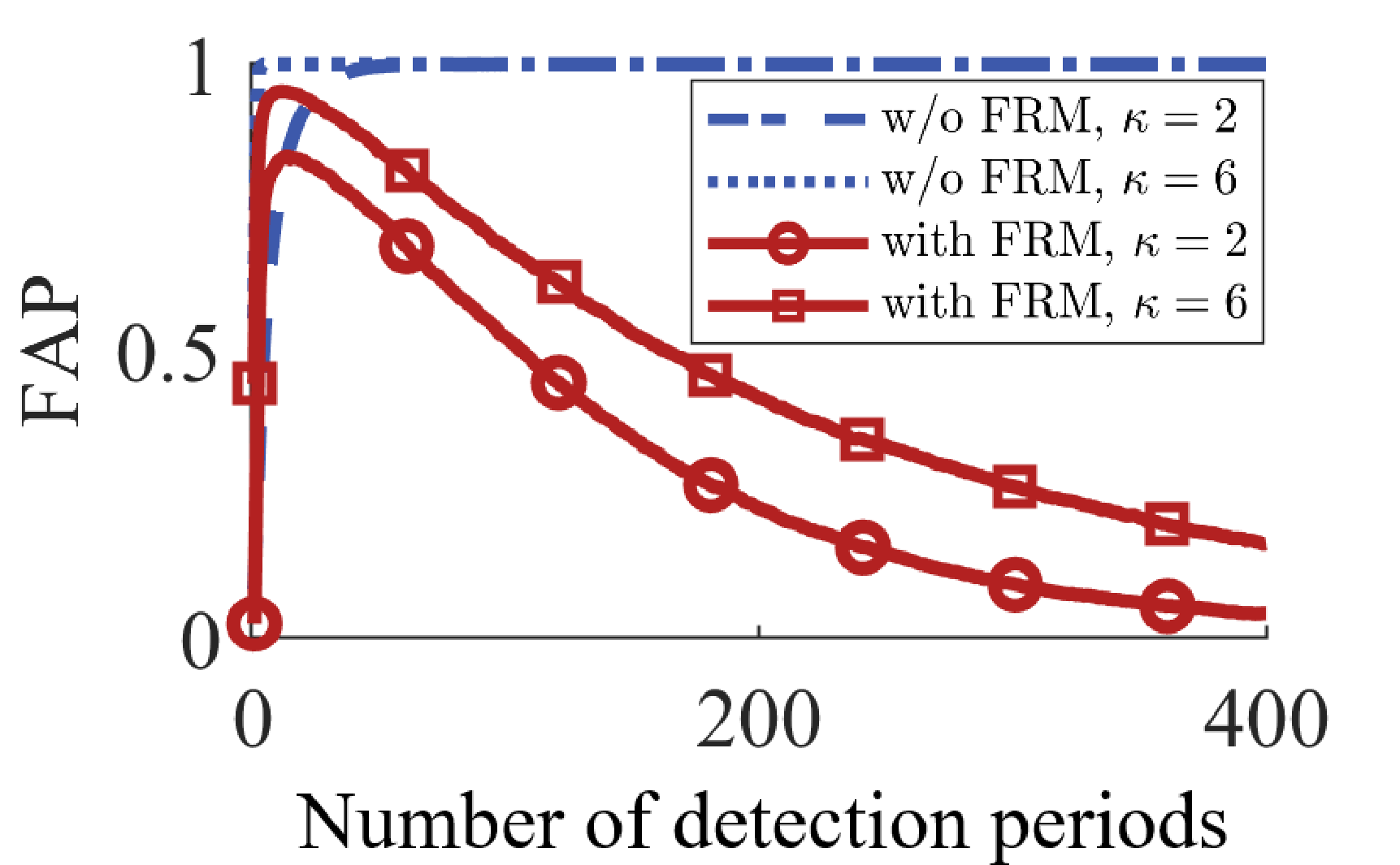}}
\subfigure[MDP (Scenorio 3)]{\label{res:root:d}\includegraphics[height=2.7cm]{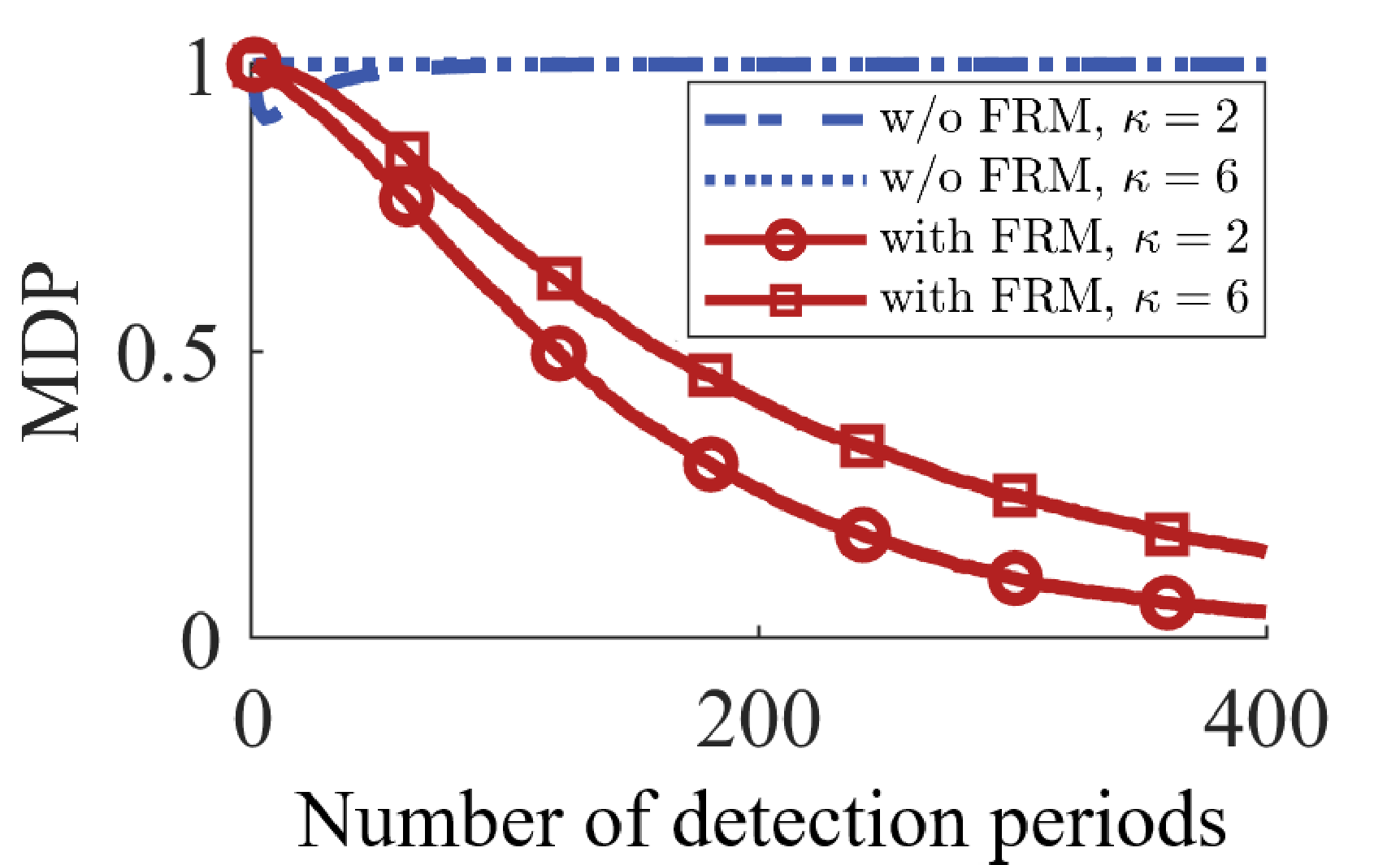}}
\subfigure[FAP (Scenorio 4)]{\label{res:root:e}\includegraphics[height=2.7cm]{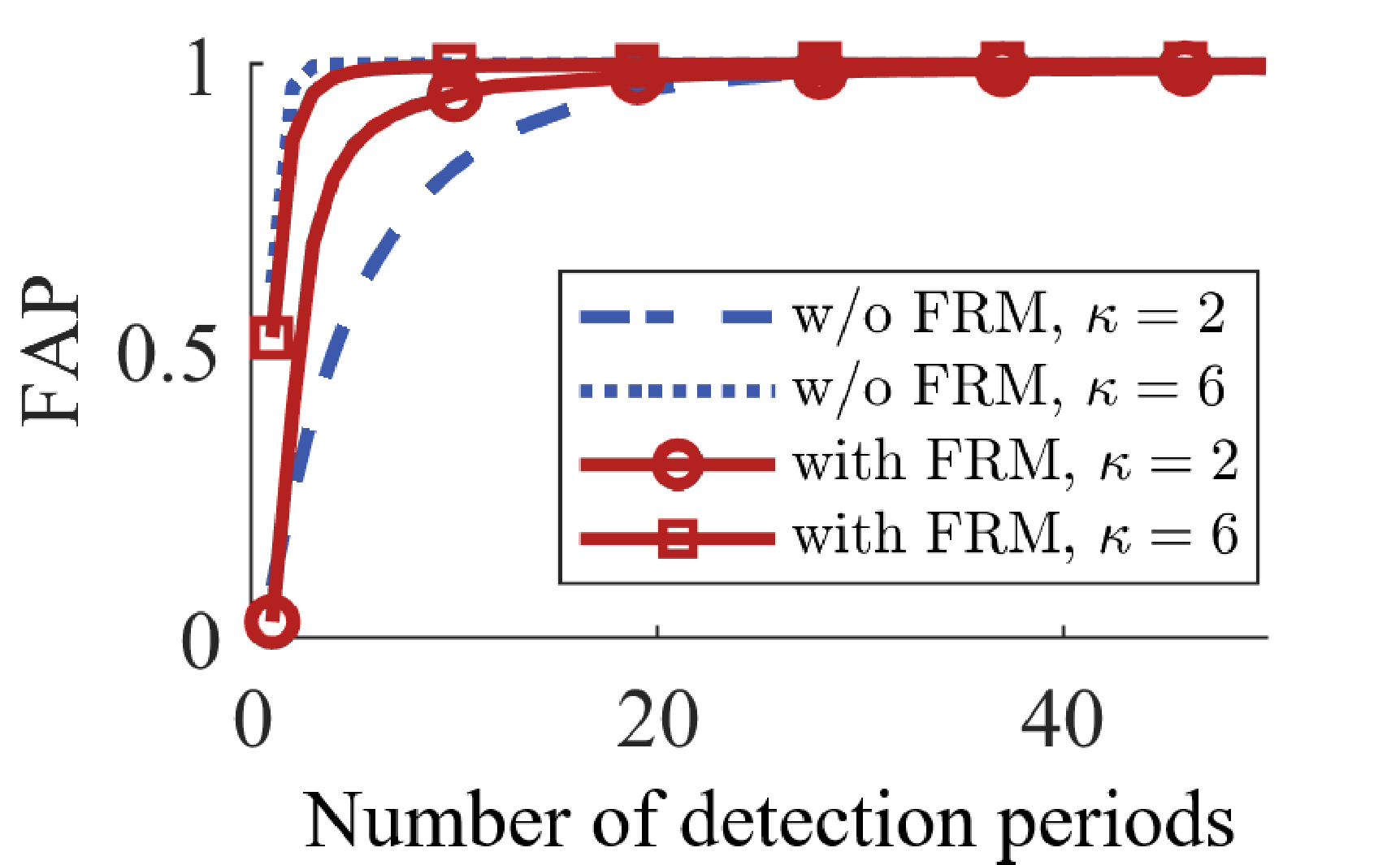}}
\subfigure[MDP (Scenorio 4)]{\label{res:root:f}\includegraphics[height=2.7cm]{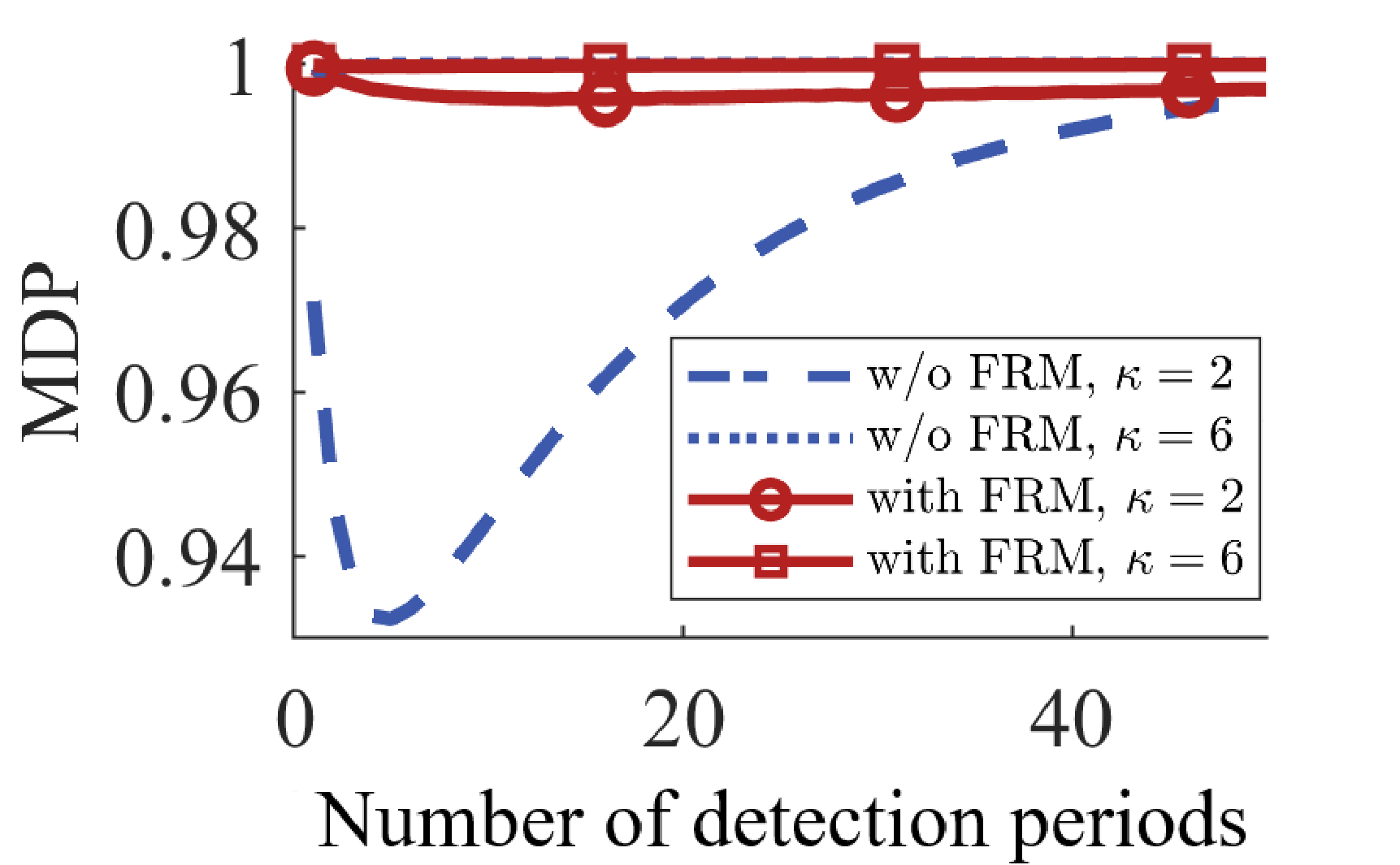}}
\caption{Performance on resisting falsified detection data. ``with FRM" -- CSD with falsification-resistance mechanism; ``w/o FRM" -- CSD without falsification-resistance mechanism. }
\vspace{-0.3cm}
\label{res:root}
\end{figure}

Fig.~\ref{res:root} presents the FAP and the MDP results in the overall detection phase for the three adversarial scenarios.
The group falsification strategy is to generate detection data with exactly opposite PLRs in null/alternative hypotheses, i.e., $q' = q_{\mathrm{d}} = 0.2$ and $q'' = q_{\mathrm{n}} = 0.15$.
The following important observations are obtained from the results.

First, the FAP and the MDP tend to 1 when the CSD does not incorporate the falsification-resistance mechanism.
Therefore, it cannot distinguish abnormal from normal nodes without removing untrustworthy LLRs, as proved in Lemma~\ref{lemma:mu}.
By contrast, the results in the first two scenarios~(Fig.~\ref{res:root:a}-\ref{res:root:d}) indicate that the proposed falsification-resistance mechanism ensures the CSD on resisting the falsified detection data and can yield up to almost $99.9\%$ improvement of the FAP and the MDP, i.e., the perfect detection is achieved.
As can be found, it would take a shorter time to achieve high detection accuracy in \texttt{Scenario 2}~(cf. Fig.~\ref{res:root:a} and Fig.~\ref{res:root:b}) than in \texttt{Scenario 3}~(cf. Fig.~\ref{res:root:c} and Fig.~\ref{res:root:d}) because there are less falsified detection data in \texttt{Scenario 2}.
However, it cannot detect anomalies in the last scenario (Fig.~\ref{res:root:e} and Fig.~\ref{res:root:f}).
This is because $\mu^{(0)}/\mu^{(1)}>1$ is not held in \texttt{Scenario 4}, i.e., the perfect detection is not satisfied as proved in Theorem~\ref{CSD:theorem3}.
These results suggest that the proposed CSD scheme is robust to the adversarial attacks due to the falsification-resistant mechanism and the LLR aggregation mechanism integrated at the trusted node, as explained in Eq.~\eqref{CSD:removed detection rule}.

\begin{table}[!tbp]
\renewcommand{\arraystretch}{1.25}
  \centering
  \caption{Detection Periods Required for a FAP Value}
  \setlength{\tabcolsep}{4mm}{
    \begin{tabular}{c|c|ccc}
    \specialrule{.08em}{0pt} {0pt}
    \multirow{2}{*}{$\kappa$} & \multirow{2}{*}{$q'$} & \multicolumn{3}{c}{Detection Periods} \\
\cline{3-5}          &       &$\text{FAP}=0.10$ & $\text{FAP}=0.05$ & $\text{FAP}=0.01$ \\
    \specialrule{.05em}{0pt} {0pt}
    \multirow{3}{*}{2} & 0.2   &   49    &   62    &  90 \\
          & 0.4   &   35    &   47    &  77\\
          & 0.8   & 34    & 46    & 76 \\
    \specialrule{.05em}{0pt} {0pt}
    \multirow{3}{*}{20} & 0.2   &   34    &    46   &  76\\
          & 0.4   &   35    &   47    &  79 \\
          & 0.8   &  36    &  48    & 80 \\
    \specialrule{.08em}{0pt} {0pt}
    \end{tabular}}
  \label{res:PFA}%
\end{table}%

\begin{table}[!t]
\renewcommand{\arraystretch}{1.25}
  \centering
  \caption{Detection Periods Required for a MDP Value}
  \setlength{\tabcolsep}{3.45mm}{
    \begin{tabular}{c|c|ccc}
    \specialrule{.08em}{0pt} {0pt}
\multirow{2}{*}{$\kappa$} & \multirow{2}{*}{$q''$} & \multicolumn{3}{c}{Detection Periods} \\
\cline{3-5}          &       &$\text{MDP}=0.10$ & $\text{MDP}=0.05$ & $\text{MDP}=0.01$ \\
    \specialrule{.05em}{0pt} {0pt}
    \multirow{3}{*}{2} & 0.05   &   65  &   81    & 112 \\
          & 0.1   &   49    &   61    &  88\\
          & 0.15   & 43    & 53    & 77 \\
    \specialrule{.05em}{0pt} {0pt}
    \multirow{3}{*}{20} & 0.05   &   36    &   49    &  75 \\
          & 0.1   &   35    &   48    &  77\\
          & 0.15   & 35    & 47    & 78 \\
    \specialrule{.08em}{0pt} {0pt}
    \end{tabular}}
    \vspace{-0.3cm}
  \label{res:PMD}%
\end{table}%

Second, the FAP and the MDP converge faster for a larger value of $\kappa$~(e.g., $\kappa=6$) in \texttt{Scenario 2}.
Conversely, a smaller value of $\kappa$ (e.g., $\kappa = 2$) results in more rapid convergence in \texttt{Scenario 3}.
The main difference between \texttt{Scenario 2} and \texttt{Scenario 3} lies in the value of $\mu^{(0)}/\mu^{(1)}$ -- the ratio of the means of packets from benign cluster nodes to malicious ones.
And, the detection efficiency of the CSD is influenced by both $\mu^{(0)}/\mu^{(1)}$ and $\kappa = \ell'(t)/\ell(t)$ -- the ratio of the falsified number to true number of packets. 
On one hand, from the condition in Eq.~\eqref{CSD:condition}, it requires fewer detection periods to reach the same performance when $\tilde{\mu}^{(0)}/\tilde{\mu}^{(1)}$ is closer to $\kappa \cdot \max \{\frac{q'- \beta}{\beta - q_{\mathrm{n}}}, \frac{\beta - q''}{q - \beta}\}$, where $\tilde{\mu}^{(0)}$ and $\tilde{\mu}^{(1)}$ are $\mu^{(0)}$ and $\mu^{(1)}$ after filtering out the cluster nodes that generate untrustworthy detection data.
Indeed, for a fixed value of $\mu^{(0)}/\mu^{(1)}$, $\tilde{\mu}^{(0)}/\tilde{\mu}^{(1)}$ grows with $\kappa$, because it is easier to find out the untrustworthy detection data.
On the other hand, both $\kappa \cdot \frac{q'-\beta}{\beta - q_{\mathrm{n}}}$ and $\kappa \cdot \frac{\beta-q''}{q - \beta} $ increase if $\kappa$ increases.
Thus, there is a trade-off between $\mu^{(0)}/\mu^{(1)}$ and $\kappa$ to determine the number of detection periods.

Finally, a series of numerical simulations is conducted to further understand the influence of different falsified PLRs on detection accuracy.
More specifically, the CSD's performance in terms of the FAP and the MDP under \texttt{Scenario 2} is analyzed by performing a grid search for the falsified PLRs $q' \in \{\, 0.2,0.4,0.8 \,\}$ and $q'' \in \{\, 0.05,0.1,0.15 \,\}$. 
The results are reported in Table~\ref{res:PFA} and Table~\ref{res:PMD}.

Table~\ref{res:PFA} shows that, when $\kappa=2$, fewer detection periods are required if $q'$ increases; when $\kappa = 20$, however, the trend is reversed.
As observed in Fig.~\ref{res:root:a} and Fig.~\ref{res:root:b}, the detection periods would be fewer for a larger $\kappa$ if $q'$ is fixed.
However, a higher value of $q'$ would reduce $\frac{\beta - q'}{q_{\mathrm{n}}-\beta}$.
Hence, there also exists a trade-off between $\kappa$ and $q'$.
Similarly, there exists another trade-off between $\kappa$ and $q''$ for MDP, which can be verified by the results presented in Table~\ref{res:PMD}.
Moreover, it also can be seen that the falsified PLR does not significantly influence the detection accuracy, especially for a larger $\kappa$, because the effect of $q''$ is neutralized by a larger value of $\kappa$.
In addition, the FAP and the MDP of the proposed robust CSD can converge to $99\%$ within a short detection time and even close to $100\%$ for a sufficiently long period.

In summary, the proposed CSD can correctly and rapidly detect abnormal packet loss and identify abnormal nodes by combining the LRT-based detection method with the modified Z-score-based falsification resistance mechanism at the trusted node, regardless of the degree of maliciousness.

\section{Conclusion}

This paper presented CSD -- a novel statistics-based framework for detecting abnormal nodes while resisting falsified detection data in clustered networks. 
A modified Z-score-based method was embedded to distinguish the falsifications. 
A comprehensive theoretical analysis has been performed to prove the effectiveness and accuracy of our scheme under the adversarial scenario where the scheme is public to attackers.
We modeled this security problem as a complete information static game and derived the optimal removal threshold of the Z-score method. 
Besides, we also proved that both the FAP and the MDP of the CSD can approach zero exponentially. 
Extensive simulations verified the superiority of our CSD.

\balance
\bibliographystyle{IEEEtran}
\bibliography{hope.bib}

\begin{thebibliography}{10}
\providecommand{\url}[1]{#1}
\csname url@samestyle\endcsname
\providecommand{\newblock}{\relax}
\providecommand{\bibinfo}[2]{#2}
\providecommand{\BIBentrySTDinterwordspacing}{\spaceskip=0pt\relax}
\providecommand{\BIBentryALTinterwordstretchfactor}{4}
\providecommand{\BIBentryALTinterwordspacing}{\spaceskip=\fontdimen2\font plus
\BIBentryALTinterwordstretchfactor\fontdimen3\font minus \fontdimen4\font\relax}
\providecommand{\BIBforeignlanguage}[2]{{%
\expandafter\ifx\csname l@#1\endcsname\relax
\typeout{** WARNING: IEEEtran.bst: No hyphenation pattern has been}%
\typeout{** loaded for the language `#1'. Using the pattern for}%
\typeout{** the default language instead.}%
\else
\language=\csname l@#1\endcsname
\fi
#2}}
\providecommand{\BIBdecl}{\relax}
\BIBdecl

\bibitem{LM23}
L.~Mamatas, V.~Demiroglou, S.~Kalafatidis, S.~Skaperas, and V.~Tsaoussidis, ``Protocol-adaptive strategies for wireless mesh smart city networks,'' \emph{{IEEE} Netw.}, vol.~37, no.~2, pp. 136--143, Apr. 2023.

\bibitem{KF23}
K.~Fizza, A.~Banerjee, P.~P. Jayaraman, N.~Auluck, R.~Ranjan, K.~Mitra, and D.~Georgakopoulos, ``A survey on evaluating the quality of autonomic internet of things applications,'' \emph{{IEEE} Commun. Surv. Tutorials}, vol.~25, no.~1, pp. 567--590, Firstquarter 2023.

\bibitem{aa21}
A.~Shahraki, A.~Taherkordi, O.~Haugen, and F.~Eliassen, ``A survey and future directions on clustering: From {WSNs} to {IoT} and modern networking paradigms,'' \emph{IEEE Trans. Netw. Serv. Manag.}, vol.~18, no.~2, pp. 2242--2274, Jun. 2021.

\bibitem{dc20}
A.~Daher, M.~Coupechoux, P.~Godlewski, P.~Ngouat, and P.~Minot, ``A dynamic clustering algorithm for multi-point transmissions in mission-critical communications,'' \emph{IEEE Trans. Wirel. Commun.}, vol.~19, no.~7, pp. 4934--4946, July 2020.

\bibitem{sp21}
A.~K. Sikder, G.~Petracca, H.~Aksu, T.~Jaeger, and A.~S. Uluagac, ``A survey on sensor-based threats and attacks to smart devices and applications,'' \emph{IEEE Commun. Surv. Tutor.}, vol.~23, no.~2, pp. 1125--1159, May 2021.

\bibitem{AB21}
M.~Al~Samara, I.~Bennis, A.~Abouaissa, and P.~Lorenz, ``An efficient outlier detection and classification clustering-based approach for {WSN},'' in \emph{Proc. IEEE Globecom}, Dec. 2021, pp. 1--6.

\bibitem{at22}
A.~Boualouache and T.~Engel, ``A survey on machine learning-based misbehavior detection systems for 5g and beyond vehicular networks,'' \emph{{IEEE} Commun. Surv. Tutorials}, vol.~25, no.~2, pp. 1128--1172, Jan. 2023.

\bibitem{sg06}
T.~Schonhoff and A.~Giordano, \emph{Detection and estimation theory and its applications}.\hskip 1em plus 0.5em minus 0.4em\relax Englewood Cliffs, NJ, USA: Prentice-Hall, 2006.

\bibitem{gs17}
F.~Gara, L.~Ben~Saad, and R.~Ben~Ayed, ``An intrusion detection system for selective forwarding attack in {IPv6}-based mobile {WSNs},'' in \emph{Proc. IWCMC}, June 2017, pp. 276--281.

\bibitem{at19}
N.~V. Abhishek, A.~Tandon, T.~J. Lim, and B.~Sikdar, ``A {GLRT}-based mechanism for detecting relay misbehavior in clustered {IoT} networks,'' \emph{IEEE Trans. Inf. Forensic Secur.}, vol.~15, pp. 435--446, 2020.

\bibitem{hz22}
Y.~Huangfu and L.~Zhou, ``Anomaly detection of wireless relays based on markov models through the wald-wolfowitz runs test,'' \emph{{IEEE} Commun. Lett.}, vol.~26, no.~11, pp. 2562--2566, Aug. 2022.

\bibitem{HZ21}
Y.~Huangfu, L.~Zhou, and F.~Zhou, ``Detecting abnormal nodes in cluster-tree networks via the likelihood ratio test of packet losses,'' in \emph{Proc. IEEE Globecom}, Dec. 2021, pp. 1--6.

\bibitem{mg00}
S.~Marti, T.~J. Giuli, K.~Lai, and M.~Baker, ``Mitigating routing misbehavior in mobile {Ad Hoc} networks,'' in \emph{Proc. ACM MobiCom}, Aug. 2000, pp. 255--265.

\bibitem{rz16}
J.~Ren, Y.~Zhang, K.~Zhang, and X.~Shen, ``Adaptive and channel-aware detection of selective forwarding attacks in wireless sensor networks,'' \emph{IEEE Trans. Wirel. Commun.}, vol.~15, no.~5, pp. 3718--3731, May 2016.

\bibitem{wl18}
E.~K. Wang, Y.~Li, Y.~Ye, S.~M. Yiu, and L.~C.~K. Hui, ``A dynamic trust framework for opportunistic mobile social networks,'' \emph{IEEE Trans. Netw. Serv. Manag.}, vol.~15, no.~1, pp. 319--329, Mar. 2018.

\bibitem{et18}
E.~Eziama, K.~Tepe, A.~Balador, K.~S. Nwizege, and L.~M.~S. Jaimes, ``Malicious node detection in vehicular {Ad-Hoc} network using machine learning and deep learning,'' in \emph{IEEE Globecom Workshops}, Dec. 2018, pp. 1--6.

\bibitem{BT23}
T.~Nguyen, J.~He, L.~T. Le, W.~Bao, and N.~H. Tran, ``Federated {PCA} on grassmann manifold for anomaly detection in iot networks,'' in \emph{{IEEE} {INFOCOM}}, May 2023, pp. 1--10.

\bibitem{ZA23}
Z.~A.~E. Houda, H.~Moudoud, B.~Brik, and L.~Khoukhi, ``Securing federated learning through blockchain and explainable {AI} for robust intrusion detection in {IoT} networks,'' in \emph{{IEEE} {INFOCOM} Workshop}, May 2023, pp. 1--6.

\bibitem{hf22}
Y.~Huangfu, L.~Zhou, and F.~Zhou, ``A new abnormal node detection method with dynamic hoeffding test,'' \emph{{IEEE} Wirel. Commun. Lett.}, vol.~11, no.~8, pp. 1595--1599, Apr. 2022.

\bibitem{kz17}
Z.~A. Khan and P.~Herrmann, ``A trust based distributed intrusion detection mechanism for internet of things,'' in \emph{Proc. IEEE AINA}, Mar. 2017, pp. 1169--1176.

\bibitem{bf19}
B.~Jedari, F.~Xia, H.~Chen, S.~K. Das, A.~Tolba, and A.-M. Zafer, ``A social-based watchdog system to detect selfish nodes in opportunistic mobile networks,'' \emph{Futur. Gener. Comp. Syst.}, vol.~92, pp. 777--788, Mar. 2019.

\bibitem{ll19}
Y.~Liu, A.~Liu, X.~Liu, and M.~Ma, ``A trust-based active detection for cyber-physical security in industrial environments,'' \emph{IEEE Trans. Ind. Inform.}, vol.~15, no.~12, pp. 6593--6603, Dec. 2019.

\bibitem{ac18}
A.~Antonopoulos and C.~Verikoukis, ``{COPS}: Cooperative statistical misbehavior mitigation in network-coding-aided wireless networks,'' \emph{IEEE Trans. Ind. Inform.}, vol.~14, no.~4, pp. 1436--1446, Apr. 2018.

\bibitem{pr20}
R.~P. Parameswarath, C.~Y. Eugene, N.~V. Abhishek, T.~J. Lim, and B.~Sikdar, ``Detecting selective forwarding using sentinels in clustered {IoT} networks,'' in \emph{Proc. IEEE Globecom}, Dec. 2020, pp. 1--6.

\bibitem{rpl12}
T.~Winter, P.~Thubert, A.~Brandt \emph{et~al.}, ``{RPL}: {IPv6} routing protocol for low-power and lossy networks.'' \emph{rfc}, vol. 6550, pp. 1--157, 2012.

\bibitem{op16}
O.~Iova, P.~Picco, T.~Istomin, and C.~Kiraly, ``{RPL}: The routing standard for the internet of things... or is it?'' \emph{IEEE Commun. Mag.}, vol.~54, no.~12, pp. 16--22, Dec. 2016.

\bibitem{tl18}
A.~Tandon, T.~J. Lim, and U.~Tefek, ``Sentinel based malicious relay detection scheme for wireless {IoT} networks,'' in \emph{Proc. IEEE Globecom Workshops}, Dec. 2018, pp. 1--6.

\bibitem{ie16}
\emph{IEEE Standard for Low-Rate Wireless Networks}, IEEE Std 802.15.4-2020 (Revision of IEEE Std 802.15.4-2015) Std., 2020.

\bibitem{zx21}
X.~Zhou, J.~Xiong, X.~Zhang, X.~Liu, and J.~Wei, ``A radio anomaly detection algorithm based on modified generative adversarial network,'' \emph{IEEE Wirel. Commun. Lett.}, vol.~10, no.~7, pp. 1552--1556, Jul. 2021.

\bibitem{LG18}
X.~Liu, Y.~Guan, and S.~W. Kim, ``Bayesian test for detecting false data injection in wireless relay networks,'' \emph{IEEE Commun. Lett.}, vol.~22, no.~2, pp. 380--383, Feb. 2018.

\bibitem{Dembo2010}
A.~Dembo and O.~Zeitouni, \emph{Large Deviations Techniques and Applications}.\hskip 1em plus 0.5em minus 0.4em\relax Berlin, Heidelberg: Springer, 2010.

\bibitem{RG92}
R.~Gibbons, \emph{A primer in game theory}.\hskip 1em plus 0.5em minus 0.4em\relax Harvester Wheatsheaf, Jun. 1992.

\bibitem{bt16}
B.~Slinker, T.~Neilands, and S.~Glantz, \emph{Primer of Applied Regression \& Analysis of Variance, Third Edition}.\hskip 1em plus 0.5em minus 0.4em\relax McGraw-Hill Education, Mar. 2016.

\bibitem{dg20}
D.~G. Zill, \emph{Advanced engineering mathematics}.\hskip 1em plus 0.5em minus 0.4em\relax Jones \& Bartlett Publishers, Dec. 2020.

\bibitem{bo09}
M.~{\'E}. Borel, ``Les probabilit{\'e}s d{\'e}nombrables et leurs applications arithm{\'e}tiques,'' \emph{Rendiconti del Circolo Matematico di Palermo}, vol.~27, no.~1, pp. 247--271, 1909.

\end{thebibliography}

\end{document}